\documentclass[onecolumn,american]{svjour3}
\usepackage[T1]{fontenc}
\usepackage[latin9]{inputenc}
\usepackage{textcomp}
\usepackage{amsmath}
\usepackage{graphicx}
\usepackage{amssymb}
\usepackage{esint}

\makeatletter

\providecommand{\tabularnewline}{\\}

\smartqed  

\makeatother

\usepackage{babel}

\begin{document}

\title{Constructing a quantum field theory from spacetime }

\author{Torsten Asselmeyer-Maluga%
\thanks{German Aerospace center, Rutherfordstr. 2, 12489 Berlin, torsten.asselmeyer-maluga@dlr.de%
} and Jerzy Kr\'ol%
\thanks{University of Silesia, Institute of Physics, ul. Uniwesytecka 4, 40-007
Katowice, iriking@wp.pl %
}}
\maketitle
\begin{abstract}
The paper shows deep connections between exotic smoothings of a small
$\mathbb{R}^{4}$ (the spacetime), the leaf space of codimension-1
foliations (related to noncommutative algebras) and quantization.
At first we relate a small exotic $\mathbb{R}^{4}$ to codimension-1
foliations of the 3-sphere unique up to foliated cobordisms and characterized
by the real-valued Godbillon-Vey invariant. Special care is taken
for the integer case which is related to flat $PSL(2,\mathbb{R})-$bundles.
Then we discuss the leaf space of the foliation using noncommutative
geometry. This leaf space contains the hyperfinite $I\! I\! I_{1}$
factor of Araki and Woods important for quantum field theory (QFT)
and the $I_{\infty}$ factor. Using Tomitas modular theory, one obtains
a relation to a factor $I\! I_{\infty}$ algebra given by the horocycle
foliation of the unit tangent bundle of a surface $S$ of genus $g>1$.
The relation to the exotic $\mathbb{R}^{4}$ is used to construct
the (classical) observable algebra as Poisson algebra of functions
over the character variety of representations of the fundamental group
$\pi_{1}(S)$ into the $SL(2,\mathbb{C})$. The Turaev-Drinfeld quantization
(as deformation quantization) of this Poisson algebra is a (complex)
skein algebra which is isomorphic to the hyperfinite factor $I\! I_{1}$
algebra determining the factor $I\! I_{\infty}=I\! I_{1}\otimes I_{\infty}$
algebra of the horocycle foliation. Therefore our geometrically motivated
hyperfinite $I\! I\! I_{1}$ factor algebra comes from the quantization
of a Poisson algebra. Finally we discuss the states and operators
to be knots and knot concordances, respectively. 
\end{abstract}
\tableofcontents{}

\section{Introduction}

The construction of quantum theories from classical theories, known
as quantization, has a long and difficult history. It starts with
the discovery of quantum mechanics in 1925 and the formalization of
the quantization procedure by Dirac and von Neumann. The construction
of a quantum theory from a given classical one is highly non-trivial
and non-unique. But except for few examples, it is the only way which
will be gone today. From a physical point of view, the world surround
us is the result of an underlying quantum theory of its constituent
parts. So, one would expect that we must understand the transition
from the quantum to the classical world. But we had developed and
tested successfully the classical theories like mechanics or electrodynamics.
Therefore one tried to construct the quantum versions out of classical
theories. In this paper we will go the other way to obtain a quantum
field theory by geometrical methods and to show its equivalence to
a quantization of a classical Poisson algebra. 

The main technical tool will be the noncommutative geometry developed
by Connes \cite{Con:85}. Then intractable space like the leaf space
of a foliation can be described by noncommutative algebras. From the
physical point of view, we have now an interpretation of noncommutative
algebras (used in quantum theory) in a geometrical context. So, we
need only an idea for the suitable geometric structure. For that purpose
one formally considers the path integral over spacetime geometries.
In the evaluation of this integral, one has to include the possibility
of different smoothness structures for spacetime \cite{Pfeiffer2004,Ass2010}.
Brans \cite{BraRan:93,Bra:94a,Bra:94b} was the first who considered
exotic smoothness also on open smooth 4-manifolds as a possibility
for space-time. He conjectured that exotic smoothness induces an additional
gravitational field (\emph{Brans conjecture}). The conjecture was
established by Asselmeyer \cite{Ass:96} in the compact case and by
S{\l{}}adkowski \cite{Sladkowski2001} in the non-compact case. S{\l{}}adkowski
\cite{Sla:96,Sla:96b,Sla:96c} discussed the influence of differential
structures on the algebra $C(M)$ of functions over the manifold $M$
with methods known as non-commutative geometry. Especially in \cite{Sla:96b,Sla:96c}
he stated a remarkable connection between the spectra of differential
operators and differential structures. But there is a big problem
which prevents progress in the understanding of exotic smoothness
especially for the $\mathbb{R}^{4}$: there is no known explicit coordinate
representation. As the result no exotic smooth function on any such
$\mathbb{R}^{4}$ is known even though there exist families of infinite
continuum many different non diffeomorphic smooth $\mathbb{R}^{4}$.
This is also a strong limitation for the applicability to physics
of non-standard open 4-smoothness. Bizaca \cite{Biz:94} was able
to construct an infinite coordinate patch by using Casson handles.
But it still seems hopeless to extract physical information from that
approach.

This situation is not satisfactory but we found a possible solution.
The solution is a careful analysis of the small exotic $\mathbb{R}^{4}$
by using foliation theory (see next section) to derive a relation
between exotic smoothness and codimension-1 foliations in section
\ref{sec:Exotic-R4-codim-1-foliation} (see Theorem \ref{thm:codim-1-foli-radial-fam}).
By using noncommutative geometry, this approach is able to produce
a von Neumann algebra via the leaf space of the foliation which can
be interpreted as the observable algebra of some QFT (see \cite{Haag:96}).
Fortunately, our approach to exotic smoothness is strongly connected
with a codimension-1 foliation of type $I\! I\! I$, i.e. the leaf
space is a factor $I\! I\! I_{1}$ von Neumann algebra. Especially
this algebra is the preferred algebra in the local algebra approach
to QFT \cite{Haag:96,Borchers2000}. Recently, this factor $I\! I\! I$
case was also discussed in connection with quantum gravity (via the
spectral triple of Connes) \cite{BertozziniConti2010}. 

In the next two sections we will give an overview about foliation
theory, its operator-theoretical description and the relation to exotic
smoothness. Both sections are rather technical with a strong overlap
to our previous paper \cite{AsselmeyerKrol2009}. In section \ref{sec:Quantization}
we turn to the quantization procedure as related to nonstandard smoothings
of $\mathbb{R}^{4}$. Based on the dictionary between operator algebra
and foliations one has the corresponding relation of small exotic
$\mathbb{R}^{4}$'s and operator algebras. This is a noncommutative
$C^{\star}$ algebra which can be seen as the algebra of quantum observables
of some theory. 
\begin{itemize}
\item First, in subsection \ref{sub:factor-III-case} we recognized the
algebra as the hyperfinite factor $I\! I\! I_{1}$ von Neumann algebra.
From Tomita-Takesaki theory it follows that any factor $I\! I\! I$
algebra $M$ decomposes as a crossed product into $M=N\rtimes_{\theta}\mathbb{R}_{+}^{*}$
where $N$ is a factor $I\! I_{\infty}$. Via Connes procedure one
can relate the factor $I\! I\! I$ foliation to a factor $I\! I$
foliation. Then we obtain a foliation of the horocycle flow on the
unit tangent bundle over some genus $g$ surface which determines
the factor $I\! I_{\infty}$. This foliation is in fact determined
by the horocycles which are closed circles.
\item Next we are looking for a classical algebraic structure which would
give the above mentioned noncommutative algebra of observables as
a result of quantization. The classical structure is recovered by
the idempotent of the $C^{\star}$ algebra and has the structure of
a Poisson algebra. The idempotents were already constructed in subsection
\ref{sub:smooth-holonomy-groupoid} as closed curves in the leaf of
the foliation of $S^{3}$. As noted by Turaev \cite{Turaev1991},
closed curves in a surface induce a Poisson algebra: Given a surface
$S$ let $X(S,G)$ be the space of flat connections of $G=SL(2,\mathbb{C})$
bundles on $S$; this space carries a Poisson structure as is shown
in subsection \ref{sub:The-observable-algebra}. The complex functions
on $X(S,SL(2,\mathbb{C}))$ can be considered as the algebra of classical
observables forming the Poisson algebra $(X(S,SL(2,\mathbb{C})),\left\{ \,,\,\right\} )$. 
\item Next in the subsection \ref{sub:Drinfeld-Turaev-Quantization} we
find a quantization procedure of the above Poisson algebra which is
the Drinfeld-Turaev deformation quantization. It is shown that the
result of this quantization is the skein algebra $(K_{t}(S),[\,,\,])$
for the deformation parameter $t=\exp(h/4)$ ($t=-1$ corresponds
to the commutative Poisson structure on $(X(S,SL(2,\mathbb{C})),\{\,,\,\})$). 
\item This skein algebra is directly related to the factor $I\! I\! I_{1}$
von Neumann algebra derived from the foliation of $S^{3}$. In fact
the skein algebra is constructed in subsection \ref{sub:Temperley-Lieb-algebra-foliation}
as the factor $I\! I_{1}$ algebra Morita equivalent to the factor
$I\! I_{\infty}$ which in turn determines the factor $I\! I\! I_{1}$
of the foliation. 
\end{itemize}
Finally in section \ref{sec:Knots-as-states} we discuss the states
of the algebra and the operators between states. Here, we present
only the ideas: the states are knots represented by holonomies along
a flat connection. Then an operator between two states is a knot concordance
(a kind of knot cobordism). The whole approach is similar to the holonomy
flux algebra of Loop quantum gravity (see \cite{Thiemann2006}). We
will discuss this interesting relation in our forthcoming work.

\section{Preliminaries: Foliations and Operator algebras}

In this section we will consider a foliation $(M,F)$ of a manifold
$M$, i.e. an integrable subbundle $F\subset TM$ of the tangent bundle
$TM$. The leaves $L$ of the foliation $(M,F)$ are the maximal connected
submanifolds $L\subset M$ with $T_{x}L=F_{x}\:\forall x\in L$. We
denote with $M/F$ the set of leaves or the leaf space. Now one can
associate to the leaf space $M/F$ a $C^{*}$algebra $C(M,F)$ by
using the smooth holonomy groupoid $G$ of the foliation (see Connes
\cite{Connes1984}). For a codimension-1 foliation of a 3-manifold
$M$ there is the Godbillon-Vey invariant \cite{GodVey:71} as element
of $H^{3}(M,\mathbb{R})$. As example we consider the construction
of a codimension-1 foliation of the 3-sphere by Thurston \cite{Thu:72}
which will be used extensively in the paper. This foliation has a
non-trivial Godbillon-Vey invariant where every element of $H^{3}(S^{3},\mathbb{R})$
is represented by a cobordism class of foliations. Hurder and Katok
\cite{HurKat:84} showed that the $C^{*}$algebra of a foliation with
non-trivial Godbillon-Vey invariant contains a factor $I\! I\! I$
subalgebra. In the following we will construct this $C^{*}$algebra
and discuss the factor $I\! I\! I$ case.

\subsection{Definition of Foliations and foliated cobordisms}

A codimension $k$ foliation%
\footnote{In general, the differentiability of a foliation is very important.
Here we consider the smooth case only. %
} of an $n$-manifold $M^{n}$ (see the nice overview article \cite{Law:74})
is a geometric structure which is formally defined by an atlas $\left\{ \phi_{i}:U_{i}\to M^{n}\right\} $,
with $U_{i}\subset\mathbb{R}^{n-k}\times\mathbb{R}^{k}$, such that
the transition functions have the form \[
\phi_{ij}(x,y)=(f(x,y),g(y)),\,\left[x\in\mathbb{R}^{n-k},y\in\mathbb{R}^{k}\right]\quad.\]
Intuitively, a foliation is a pattern of $(n-k)$-dimensional stripes
- i.e., submanifolds - on $M^{n}$, called the leaves of the foliation,
which are locally well-behaved. The tangent space to the leaves of
a foliation $\mathcal{F}$ forms a vector bundle over $M^{n}$, denoted
$T\mathcal{F}$. The complementary bundle $\nu\mathcal{F}=TM^{n}/T\mathcal{F}$
is the normal bundle of $\mathcal{F}$. Such foliations are called
regular in contrast to singular foliations or Haefliger structures.
For the important case of a codimension-1 foliation we need an overall
non-vanishing vector field or its dual, an one-form $\omega$. This
one-form defines a foliation iff it is integrable, i.e.\[
d\omega\wedge\omega=0\]
and the leaves are the solutions of the equation $\omega=const.$

Now we will discuss an important equivalence relation between foliations,
cobordant foliations. Let $M_{0}$ and $M_{1}$ be two closed, oriented
$m$-manifolds with codimension-$q$ foliations. Then these foliated
manifolds are said to be \emph{foliated cobordant} if there is a compact,
oriented $(m+1)$-manifold with boundary $\partial W=M_{0}\sqcup\overline{M}_{1}$
and with a codimension-$q$ foliation transverse to the boundary and
inducing the given foliation there. The resulting foliated cobordism
classes form a group under disjoint union.

\subsection{Non-cobordant foliations of $S^{3}$ detected by the Godbillon-Vey
class \label{sub:Non-cobordant-foliations-S3}}

In \cite{Thu:72}, Thurston constructed a foliation of the 3-sphere
$S^{3}$ depending on a polygon $P$ in the hyperbolic plane $\mathbb{H}^{2}$
so that two foliations are non-cobordant if the corresponding polygons
have different areas. For later usage, we will present this construction
now (see also the book \cite{Tamura1992} chapter VIII for the details).

Consider the hyperbolic plane $\mathbb{H}^{2}$ and its unit tangent
bundle $T_{1}\mathbb{H}^{2}$ , i.e the tangent bundle $T\mathbb{H}^{2}$
where every vector in the fiber has norm $1$. Thus the bundle $T_{1}\mathbb{H}^{2}$
is a $S^{1}$-bundle over $\mathbb{H}^{2}$. There is a foliation
$\mathcal{F}$ of $T_{1}\mathbb{H}^{2}$ invariant under the isometries
of $\mathbb{H}^{2}$ which is induced by bundle structure and by a
family of parallel geodesics on $\mathbb{H}^{2}$. The foliation $\mathcal{F}$
is transverse to the fibers of $T_{1}\mathbb{H}^{2}$. Let $P$ be
any convex polygon in $\mathbb{H}^{2}$. We will construct a foliation
$\mathcal{F}_{P}$ of the three-sphere $S^{3}$ depending on $P$.
Let the sides of $P$ be labeled $s_{1},\ldots,s_{k}$ and let the
angles have magnitudes $\alpha_{1},\ldots,\alpha_{k}$. Let $Q$ be
the closed region bounded by $P\cup P'$, where $P'$ is the reflection
of $P$ through $s_{1}$. Let $Q_{\epsilon}$, be $Q$ minus an open
$\epsilon$-disk about each vertex. If $\pi:T_{1}\mathbb{H}^{2}\to\mathbb{H}^{2}$
is the projection of the bundle $T_{1}\mathbb{H}^{2}$, then $\pi^{-1}(Q)$
is a solid torus $Q\times S^{1}$(with edges) with foliation $\mathcal{F}_{1}$
induced from $\mathcal{F}$. For each $i$, there is an unique orientation-preserving
isometry of $\mathbb{H}^{2}$, denoted $I_{i}$, which matches $s_{i}$
point-for-point with its reflected image $s'_{i}$. We glue the cylinder
$\pi^{-1}(s_{i}\cap Q_{\epsilon})$ to the cylinder $\pi^{-1}(s'_{i}\cap Q_{\epsilon})$
by the differential $dI_{i}$ for each $i>1$, to obtain a manifold
$M=(S^{2}\setminus\left\{ \mbox{\mbox{k} punctures}\right\} )\times S^{1}$,
and a (glued) foliation $\mathcal{F}_{2}$, induced from $\mathcal{F}_{1}$.
To get a complete $S^{3}$, we have to glue-in $k$ solid tori for
the $k$ $S^{1}\times\mbox{punctures}.$ Now we choose a linear foliation
of the solid torus with slope $\alpha_{k}/\pi$ (Reeb foliation).
Finally we obtain a smooth codimension-1 foliation $\mathcal{F}_{P}$
of the 3-sphere $S^{3}$ depending on the polygon $P$.

Now we consider two codimension-1 foliations $\mathcal{F}_{1},\mathcal{F}_{2}$
depending on the convex polygons $P_{1}$ and $P_{2}$ in $\mathbb{H}^{2}$.
As mentioned above, these foliations $\mathcal{F}_{1},\mathcal{F}_{2}$
are defined by two one-forms $\omega_{1}$ and $\omega_{2}$ with
$d\omega_{a}\wedge\omega_{a}=0$ and $a=0,1$. Now we define the one-forms
$\theta_{a}$ as the solution of the equation\[
d\omega_{a}=-\theta_{a}\wedge\omega_{a}\]
and consider the closed 3-form\begin{equation}
\Gamma_{\mathcal{F}_{a}}=\theta_{a}\wedge d\theta_{a}\label{eq:Godbillon-Vey-class}\end{equation}
 associated to the foliation $\mathcal{F}_{a}$. As discovered by
Godbillon and Vey \cite{GodVey:71}, $\Gamma_{\mathcal{F}}$ depends
only on the foliation $\mathcal{F}$ and not on the realization via
$\omega,\theta$. Thus $\Gamma_{\mathcal{F}}$, the \emph{Godbillon-Vey
class}, is an invariant of the foliation. Let $\mathcal{F}_{1}$ and
$\mathcal{F}_{2}$ be two cobordant foliations then $\Gamma_{\mathcal{F}_{1}}=\Gamma_{\mathcal{F}_{2}}$.
In case of the polygon-dependent foliations $\mathcal{F}_{1},\mathcal{F}_{2}$,
Thurston \cite{Thu:72} obtains\[
\Gamma_{\mathcal{F}_{a}}=vol(\pi^{-1}(Q))=4\pi\cdot Area(P_{a})\]
and thus
\begin{itemize}
\item $\mathcal{F}_{1}$ is cobordant to $\mathcal{F}_{2}$ $\Longrightarrow$$Area(P_{1})=Area(P_{2})$
\item $\mathcal{F}_{1}$ and $\mathcal{F}_{2}$ are non-cobordant $\Longleftrightarrow$$Area(P_{1})\not=Area(P_{2})$
\end{itemize}
We note that $Area(P)=(k-2)\pi-\sum_{k}\alpha_{k}$. The Godbillon-Vey
class is an element of the deRham cohomology $H^{3}(S^{3},\mathbb{R})$
which will be used later to construct a relation to gerbes. Furthermore
we remark that the classification is not complete. Thurston constructed
only a surjective homomorphism from the group of cobordism classes
of foliation of $S^{3}$ into the real numbers $\mathbb{R}$. We remark
the close connection between the Godbillon-Vey class (\ref{eq:Godbillon-Vey-class})
and the Chern-Simons form if $\theta$ can be interpreted as connection
of a suitable line bundle.

\subsection{Codimension-one foliations on 3-manifolds}

Now we will discuss the general case of a compact 3-manifold. Later
on we will need the codimension-1 foliations of a homology 3-sphere
$\Sigma$. Because of the diffeomorphism $\Sigma\#S^{3}=\Sigma$,
we can relate a foliation on $\Sigma$ to a foliation on $S^{3}$.
By using the surgery along a knot or link, we are able to construct
the codimension-one foliation for every compact 3-manifold.
\begin{theorem}
\label{thm:foliation-3MF}Given a compact 3-manifold $\Sigma$ without
boundary. Every codimension-one foliation $\mathcal{F}$ of the 3-sphere
$S^{3}$ (constructed above) induces a codimension-one foliation $\mathcal{F}_{\Sigma}$
on $\Sigma$. For every cobordism class $[\mathcal{F}]$ as element
of the deRham cohomology $H^{3}(S^{3},\mathbb{R})$, there exists
an element of $H^{3}(\Sigma,\mathbb{R})$ with a cobordism class $[\mathcal{F}_{\Sigma}]$.\end{theorem}
\begin{proof}
The proof can be found in \cite{AsselmeyerKrol2009}. $\square$
\end{proof}

\subsection{The smooth holonomy groupoid and its $C^{*}$algebra\label{sub:smooth-holonomy-groupoid}}

Let $(M,F)$ be a foliated manifold. Now we shall construct a von
Neumann algebra $W(M,F)$ canonically associated to $(M,F)$ and depending
only on the Lebesgue measure class on the space $X=M/F$ of leaves
of the foliation. \emph{In the following we will identify the leaf
space with this von Neumann algebra.} The classical point of view,
$L^{\infty}(X)$, will only give the center $Z(W)$ of $W$. According
to Connes \cite{Connes94}, we assign to each leaf $\ell\in X$ the
canonical Hilbert space of square-integrable half-densities $L^{2}(\ell)$.
This assignment, i.e. a measurable map, is called a random operator
forming a von Neumann $W(M,F)$. The explicit construction of this
algebra can be found in \cite{Connes1984}. Here we remark that $W(M,F)$
is also a noncommutative Banach algebra which is used above. Alternatively
we can construct $W(M,F)$ as the compact endomorphisms of modules
over the $C^{*}$ algebra $C^{*}(M,F)$ of the foliation $(M,F)$
also known as holonomy algebra. From the point of view of K theory,
both algebras $W(M,F)$ and $C^{*}(M,F)$ are Morita-equivalent to
each other leading to the same $K$ groups. In the following we will
construct the algebra $C^{*}(M,F)$ by using the holonomy groupoid
of the foliation.

Given a leaf $\ell$ of $(M,F)$ and two points $x,y\in\ell$ of this
leaf, any simple path $\gamma$ from $x$ to $y$ on the leaf $\ell$
uniquely determines a germ $h(\gamma)$ of a diffeomorphism from a
transverse neighborhood of $x$ to a transverse neighborhood of $y$.
The germ of diffeomorphism $h(\gamma)$ thus obtained only depends
upon the homotopy class of $\gamma$ in the fundamental groupoid of
the leaf $\ell$, and is called the holonomy of the path $\gamma$.
The holonomy groupoid of a leaf $\ell$ is the quotient of its fundamental
groupoid by the equivalence relation which identifies two paths $\gamma$
and $\gamma'$ from $x$ to $y$ (both in $\ell$) iff $h(\gamma)=h(\gamma')$.
The holonomy covering $\tilde{\ell}$ of a leaf is the covering of
$\ell$ associated to the normal subgroup of its fundamental group
$\pi_{1}(\ell)$ given by paths with trivial holonomy. The holonomy
groupoid of the foliation is the union $G$ of the holonomy groupoids
of its leaves. 

Recall a groupoid ${\tt G}$ is a category where every morphism is
invertible. Let $G_{0}$ be a set of objects and $G_{1}$ the set
of morphisms of ${\tt G}$, then the structure maps of ${\tt G}$
reads as:

\begin{equation}
G_{1}\,_{t}\times_{s}G_{1}\overset{m}{\rightarrow}G_{1}\overset{i}{\rightarrow}G_{1}\overset{s}{\underset{t}{\rightrightarrows}}G_{0}\overset{e}{\rightarrow}G_{1}\label{eq:def-groupoid}\end{equation}
where $m$ is the composition of the composable two morphisms (target
of the first is the source of the second), $i$ is the inversion of
an arrow, $s,\, t$ the source and target maps respectively, $e$
assigns the identity to every object. We assume that $G_{0,1}$ are
smooth manifolds and all structure maps are smooth too. We require
that the $s,\, t$ maps are submersions, thus $G_{1}\,_{t}\times_{s}G_{1}$
is a manifold as well. These groupoids are called \emph{smooth} groupoids.

Given an element $\gamma$ of $G$, we denote by $x=s(\gamma)$ the
origin of the path $\gamma$ and its endpoint $y=t(\gamma)$ with
the range and source maps $t,s$. An element $\gamma$ of $G$ is
thus given by two points $x=s(\gamma)$ and $y=r(\gamma)$ of $M$
together with an equivalence class of smooth paths: the $\gamma(t)$,
$t\in[0,1]$with $\gamma(0)=x$ and $\gamma(1)=y$, tangent to the
bundle $F$ (i.e. with $\frac{d}{dt}\gamma(t)\in F_{\gamma(t)}$,
$\forall t\in\mathbb{R}$) identifying $\gamma_{1}$ and $\gamma_{2}$
as equivalent iff the holonomy of the path $\gamma_{2}\circ\gamma_{1}^{-1}$
at the point $x$ is the identity. The graph $G$ has an obvious composition
law. For $\gamma,\gamma'\in G$ , the composition $\gamma\circ\gamma'$
makes sense if $s(\gamma)=t(\gamma)$. The groupoid $G$ is by construction
a (not necessarily Hausdorff) manifold of dimension $\dim G=\dim V+\dim F$.
We state that $G$ is a smooth groupoid, the \emph{smooth holonomy
groupoid}.

Then the $C^{*}$algebra $C_{r}^{*}(M,F)$ of the foliation $(M,F)$
is the $C^{*}$algebra $C_{r}^{*}(G)$ of the smooth holonomy groupoid
$G$. For completeness we will present the explicit construction (see
\cite{Connes94} sec. II.8). The basic elements of $C_{r}^{*}(M,F)$)
are smooth half-densities with compact supports on $G$, $f\in C_{c}^{\infty}(G,\Omega^{1/2})$,
where $\Omega_{\gamma}^{1/2}$ for $\gamma\in G$ is the one-dimensional
complex vector space $\Omega_{x}^{1/2}\otimes\Omega_{y}^{1/2}$, where
$s(\gamma)=x,t(\gamma)=y$, and $\Omega_{x}^{1/2}$ is the one-dimensional
complex vector space of maps from the exterior power $\Lambda^{k}F_{x}$
,$k=\dim F$, to $\mathbb{C}$ such that \[
\rho(\lambda\nu)=|\lambda|^{1/2}\rho(\nu)\qquad\forall\nu\in\Lambda^{k}F_{x},\lambda\in\mathbb{R}\:.\]
For $f,g\in C_{c}^{\infty}(G,\Omega^{1/2})$, the convolution product
$f*g$ is given by the equality\[
(f*g)(\gamma)=\intop_{\gamma_{1}\circ\gamma_{2}=\gamma}f(\gamma_{1})g(\gamma_{2})\]
Then we define via $f^{*}(\gamma)=\overline{f(\gamma^{-1})}$ a $*$operation
making $C_{c}^{\infty}(G,\Omega^{1/2})$ into a $*$algebra. For each
leaf $L$ of $(M,F)$ one has a natural representation of $C_{c}^{\infty}(G,\Omega^{1/2})$
on the $L^{2}$ space of the holonomy covering $\tilde{L}$ of $L$.
Fixing a base point $x\in L$, one identifies $\tilde{L}$ with $G_{x}=\left\{ \gamma\in G,\, s(\gamma)=x\right\} $and
defines the representation\[
(\pi_{x}(f)\xi)(\gamma)=\intop_{\gamma_{1}\circ\gamma_{2}=\gamma}f(\gamma_{1})\xi(\gamma_{2})\qquad\forall\xi\in L^{2}(G_{x}).\]
The completion of $C_{c}^{\infty}(G,\Omega^{1/2})$ with respect to
the norm \[
||f||=\sup_{x\in M}||\pi_{x}(f)||\]
makes it into a $C^{*}$algebra $C_{r}^{*}(M,F)$. Among all elements
of the $C^{*}$ algebra, there are distinguished elements, idempotent
operators or projectors having a geometric interpretation in the foliation.
For later use, we will construct them explicitly (we follow \cite{Connes94}
sec. $II.8.\beta$ closely). Let $N\subset M$ be a compact submanifold
which is everywhere transverse to the foliation (thus $\dim(N)=\mathrm{codim}(F)$).
A small tubular neighborhood $N'$ of $N$ in $M$ defines an induced
foliation $F'$ of $N'$ over $N$ with fibers $\mathbb{R}^{k},\, k=\dim F$.
The corresponding $C^{*}$algebra $C_{r}^{*}(N',F')$ is isomorphic
to $C(N)\otimes\mathcal{K}$ with $\mathcal{K}$ the $C^{*}$ algebra
of compact operators. In particular it contains an idempotent $e=e^{2}=e^{*}$,
$e=1_{N}\otimes f\in C(N)\otimes\mathcal{K}$ , where $f$ is a minimal
projection in $\mathcal{K}$. The inclusion $C_{r}^{*}(N',F')\subset C_{r}^{*}(M,F)$
induces an idempotent in $C_{r}^{*}(M,F)$. Now we consider the range
map $t$ of the smooth holonomy groupoid $G$ defining via $t^{-1}(N)\subset G$
a submanifold. Let $\xi\in C_{c}^{\infty}(t^{-1}(N),s^{*}(\Omega^{1/2}))$
be a section (with compact support) of the bundle of half-density
$s^{*}(\Omega^{1/2})$ over $t^{-1}(N)$ so that the support of $\xi$
is in the diagonal in $G$ and \[
\intop_{t(\gamma)=y}|\xi(\gamma)|^{2}=1\qquad\forall y\in N.\]
Then the equality\[
e(\gamma)=\sum_{{s(\gamma)=s(\gamma')\atop t(\gamma')\in N}}\bar{\xi}(\gamma'\circ\gamma^{-1})\xi(\gamma')\]
defines an idempotent $e\in C_{c}^{\infty}(G,\Omega^{1/2})\subset C_{r}^{*}(M,F)$.
Thus, such an idempotent is given by a closed curve in $M$ transversal
to the foliation.

\subsection{Some information about the factor $I\! I\! I$ case\label{sub:factor-III-case}}

In our case of codimension-1 foliations of the 3-sphere with nontrivial
Godbillon-Vey invariant we have the result of Hurder and Katok \cite{HurKat:84}.
Then the corresponding von Neumann algebra $W(S^{3},F)$ contains
a factor $I\! I\! I$ algebra. At first we will give an overview about
the factor $I\! I\! I$.

Remember a von Neumann algebra is an involutive subalgebra $M$ of
the algebra of operators on a Hilbert space $H$ that has the property
of being the commutant of its commutant: $(M')'=M$. This property
is equivalent to saying that $M$ is an involutive algebra of operators
that is closed under weak limits. A von Neumann algebra $M$ is said
to be hyperfinite if it is generated by an increasing sequence of
finite-dimensional subalgebras. Furthermore we call $M$ a factor
if its center is equal to $\mathbb{C}$. It is a deep result of Murray
and von Neumann that every factor $M$ can be decomposed into 3 types
of factors $M=M_{I}\oplus M_{II}\oplus M_{III}$. The factor $I$
case divides into the two classes $I_{n}$ and $I_{\infty}$ with
the hyperfinite factors $I_{n}=M_{n}(\mathbb{C})$ the complex square
matrices and $I_{\infty}=\mathcal{L}(H)$ the algebra of all operators
on an infinite-dimensional Hilbert space $H$. The hyperfinite $I\! I$
factors are given by $I\! I_{1}=Cliff_{\mathbb{C}}(E)$, the Clifford
algebra of an infinite-dimensional Euclidean space $E$, and $I\! I_{\infty}=I\! I_{1}\otimes I_{\infty}$.
The case $I\! I\! I$ remained mysterious for a long time. Now we
know that there are three cases parametrized by a real number $\lambda\in[0,1]$:
$I\! I\! I_{0}=R_{W}$ the Krieger factor induced by an ergodic flow
$W$, $I\! I\! I_{\lambda}=R_{\lambda}$ the Powers factor for $\lambda\in(0,1)$
and $I\! I\! I_{1}=R_{\infty}=R_{\lambda_{1}}\otimes R_{\lambda_{2}}$
the Araki-Woods factor for all $\lambda_{1},\lambda_{2}$ with $\lambda_{1}/\lambda_{2}\notin\mathbb{Q}$.
We remark that all factor  $I\! I\! I$ cases are induced by infinite
tensor products of the other factors. One example of such an infinite
tensor space is the Fock space in quantum field theory. 

But now we are interested in an explicit construction of a factor
$I\! I\! I$ von Neumann algebra of a foliation. The interesting example
of this situation is given by the Anosov foliation $F$ of the unit
sphere bundle $V=T_{1}S$ of a compact Riemann surface $S$ of genus
$g>1$ endowed with its Riemannian metric of constant curvature $-1$.
In general the manifold $V$ is the quotient $V=G/T$ of the semi-simple
Lie group $G=PSL(2,\mathbb{R})$, the isometry group of the hyperbolic
plane $\mathbb{H}^{2}$, by the discrete cocompact subgroup $T=\pi_{1}(S)$,
and the foliation $F$ of $V$ is given by the orbits of the action
by left multiplication on $V=G/T$ of the subgroup of upper triangular
matrices of the form\[
\left(\begin{array}{cc}
1 & t\\
0 & 1\end{array}\right)\quad t\in\mathbb{R}\]
The von Neumann algebra $M=W(V,F)$ of this foliation is the (unique)
hyperfinite factor of type $I\! I\! I_{1}=R_{\infty}$. In the subsection
\ref{sub:Non-cobordant-foliations-S3} we describe the construction
of the codimension-1 foliation on the 3-sphere $S^{3}$. The main
ingredient of this construction is the convex polygon $P$ in the
hyperbolic plane $\mathbb{H}^{2}$ having curvature $-1$. The Reeb
components of this foliation of $S^{3}$ are represented by a factor
$I_{\infty}$ algebra and thus do not contribute to the Godbillon-Vey
class. Putting all things together we will get 
\begin{theorem}
The codimension-1 foliation of the 3-sphere $S^{3}$ with non-trivial
Godbillon-Vey invariant is also associated to a von Neumann algebra
$W(S^{3},F)$ induced by the foliation which contains a factor  $I\! I\! I$
algebra, the hyperfinite $I\! I\! I_{1}$ factor $R_{\infty}$. \end{theorem}
\begin{proof}
This theorem follows mostly from the work Hurder and Katok \cite{HurKat:84}.
The codimension-1 foliation of the 3-sphere was constructed in subsection
\ref{sub:Non-cobordant-foliations-S3}. It admits a non-trivial Godbillon-Vey
invariant related to the volume of the polygon $P$ in $\mathbb{H}^{2}$.
The whole construction do not depend on the number of vertices of
$P$ but on the volume $vol(P)$ only. Thus without loss of generality,
we can choose the even number $4g$ for $g\in\mathbb{N}$ of vertices
for $P$. As model of the hyperbolic plane we choose the usual upper
half-plane model where the group $SL(2,\mathbb{R})$ (the real M\"obius
transformations) and the hyperbolic group $PSL(2,\mathbb{R})$ (the
group of all orientation-preserving isometries of $\mathbb{H}^{2}$)
act via fractional linear transformations. Then the polygon $P$ is
a fundamental polygon representing a Riemann surface $S$ of genus
$g$. Via the procedure above, we can construct a foliation on $T_{1}S=PSL(2,\mathbb{R})/T$
with $T=\pi_{1}(S)$. This foliation is also induced from the foliation
of $T_{1}\mathbb{H}$(as well as the foliation of the $S^{3}$) via
the left action above. The difference between the foliation on $T_{1}S$
and on $S^{3}$ is given by the different usage of the polygon $P$.
Thus the von Neumann algebra $W(S^{3},F)$ of the codimension-1 foliation
of the 3-sphere contains a factor  $I\! I\! I$ algebra in agreement
with the results in \cite{HurKat:84}. In the notation above we have
the unit tangent bundle $T_{1}P$ of the polygon $P$ equipped with
an Anosov foliation (see also \cite{Tamura1992}). The group $PSL(2,\mathbb{R})$
acts as isometry on $\mathbb{H}^{2}$ where the modular group $PSL(2,\mathbb{Z})$
acts as discrete subgroup leaving the polygon $P$ (seen as fundamental
domain) invariant. The upper triangular matrices above are elements
of $PSL(2,\mathbb{R})$ and act by linear fractional transformation
inducing a shift. The orbits of this action have therefore constant
velocity (the horocycle flow) and we are done. $\square$

We showed that this factor  $I\! I\! I$ algebra is the hyperfinite
$I\! I\! I_{1}$ factor $R_{\infty}$. Now one may ask, what is the
physical meaning of the factor $I\! I\! I$? Because of the Tomita-Takesaki-theory,
factor $I\! I\! I$ algebras are deeply connected to the characterization
of equilibrium temperature states of quantum states in statistical
mechanics and field theory also known as Kubo-Martin-Schwinger (KMS)
condition. Furthermore in the quantum field theory with local observables
(see Borchers \cite{Borchers2000} for an overview) one obtains close
connections to Tomita-Takesaki-theory. For instance one was able to
show that on the vacuum Hilbert space with one vacuum vector the algebra
of local observables is a factor $I\! I\! I_{1}$ algebra. As shown
by Thiemann et. al. \cite{Thiemann2006} on a class of diffeomorphism
invariant theories there exists an unique vacuum vector. Thus the
observables algebra must be of this type.
\end{proof}

\section{Exotic $\mathbb{R}^{4}$ and codimension-one foliations\label{sec:Exotic-R4-codim-1-foliation}}

Einsteins insight that gravity is the manifestation of geometry leads
to a new view on the structure of spacetime. From the mathematical
point of view, spacetime is a smooth 4-manifold endowed with a (smooth)
metric as basic variable for general relativity. Later on, the existence
question for Lorentz structure and causality problems (see Hawking
and Ellis \cite{HawEll:94}) gave further restrictions on the 4-manifold:
causality implies non-compactness, Lorentz structure needs a codimension-1
foliation. Usually, one starts with a globally foliated, non-compact
4-manifold $\Sigma\times\mathbb{R}$ fulfilling all restrictions where
$\Sigma$ is a smooth 3-manifold representing the spatial part. But
other non-compact 4-manifolds are also possible, i.e. it is enough
to assume a non-compact, smooth 4-manifold endowed with a codimension-1
foliation. All these restrictions on the representation of spacetime
by the manifold concept are clearly motivated by physical questions.
Among the properties there is one distinguished element: the smoothness.
Usually one assumes a smooth, unique atlas of charts (i.e. a smooth
or differential structure) covering the manifold where the smoothness
is induced by the unique smooth structure on $\mathbb{R}$. But that
is not the full story. Even in dimension 4, there are an infinity
of possible other smoothness structures (i.e. a smooth atlas) non-diffeomorphic
to each other. For a deeper insight we refer to the book \cite{Asselmeyer2007}.

\subsection{Smoothness on manifolds}

If two manifolds are homeomorphic but non-diffeomorphic, they are
\textbf{exotic} to each other. The smoothness structure is called
an \textbf{exotic smoothness structure}.

The implications for physics are tremendous because we rely on the
smooth calculus to formulate field theories. Thus different smoothness
structures have to represent different physical situations leading
to different measurable results. But it should be stressed that \emph{exotic
smoothness is not exotic physics.} Exotic smoothness is a mathematical
structure which should be further explored to understand its physical
relevance.

Usually one starts with a topological manifold $M$ and introduces
structures on them. Then one has the following ladder of possible
structures:\begin{eqnarray*}
\mbox{Topology}\to & \mbox{\mbox{piecewise-linear(PL)}}\to & \mbox{Smoothness}\to\\
\qquad\to & \mbox{bundles, Lorentz, Spin etc.}\to & \mbox{metric, geometry,...}\end{eqnarray*}
We do not want to discuss the first transition, i.e. the existence
of a triangulation on a topological manifold. But we remark that the
existence of a PL structure implies uniquely a smoothness structure
in all dimensions smaller than 7 \cite{KirSie:77}. The following
basic facts should the reader keep in mind for any $n-$dimensional
manifold $M^{n}$:
\begin{enumerate}
\item The maximal differentiable atlas $\mathcal{A}$ of $M^{n}$ is the
smoothness structure.
\item Every manifold $M^{n}$ can be embedded in $\mathbb{R}^{N}$ with
$N>2n$. A smooth embedding $M^{n}\hookrightarrow\mathbb{R}^{N}$
induces the \emph{standard smooth structure} on $M$. All other possible
smoothness structures are called exotic smoothness structures.
\item The existence of a smoothness structure is \emph{necessary} to introduce
Riemannian or Lorentz structures on $M$, but the smoothness structure
don't further restrict the Lorentz structure.
\end{enumerate}

\subsection{Small exotic $\mathbb{R}^{4}$'s and Akbulut corks}

Now we consider two homeomorphic, smooth, but non-diffeomorphic 4-manifolds
$M_{0}$ and $M$. As expressed above, a comparison of both smoothness
structures is given by a h-cobordism $W$ between $M_{0}$ and $M$
($M,M_{0}$ are homeomorphic). Let the 4-manifolds additionally be
compact, closed and simple-connected, then we have the structure theorem%
\footnote{A diffeomorphism will be described by the symbol $=$ in the following.%
} of h-cobordisms \cite{CuFrHsSt:97}:
\begin{theorem}
\label{thm:Akbulut-cork}Let $W$ be a h-cobordisms between $M_{0},\, M$,
then there are contractable submanifolds $A_{0}\subset M_{0},\,\, A\subset M$
and a h subcobordism $X\subset W$ with $\partial X=A_{0}\sqcup A$,
so that the remaining h-cobordism $W\setminus X$ trivializes $W\setminus X=(M_{0}\setminus A_{0})\times[0,1]$
inducing a diffeomorphism between $M_{0}\setminus A_{0}$ and $M\setminus A$.
\end{theorem}
In short it means that the smoothness structure of $M$ is determined
by the contractable manifold $A$ -- \emph{its} \emph{Akbulut cork}
-- and by the embedding of $A$ into $M$. As shown by Freedman\cite{Fre:82},
the Akbulut cork has a homology 3-sphere%
\footnote{A homology 3-sphere is a 3-manifold with the same homology as the
3-sphere $S^{3}$.%
} as boundary. The embedding of the cork can be derived now from the
structure of the h-subcobordism $X$ between $A_{0}$ and $A$. For
that purpose we cut $A_{0}$ out from $M_{0}$ and $A$ out from $M$.
Then we glue in both submanifolds $A_{0},A$ via the maps $\tau_{0}:\partial A_{0}\to\partial(M_{0}\setminus A_{0})=\partial A_{0}$
and $\tau:\partial A\to\partial(M\setminus A)=\partial A$. Both maps
$\tau_{0},\tau$ are involutions, i.e. $\tau\circ\tau=id$. One of
these maps (say $\tau_{0}$) can be chosen to be trivial (say $\tau_{0}=id$).
Thus the involution $\tau$ determines the smoothness structure. Especially
the topology of the Akbulut cork $A$ and its boundary $\partial A$
is given by the topology of $M$. For instance, the Akbulut cork of
the blow-uped 4-dimensional \emph{K3 surface $K3\#\overline{\mathbb{C}P}^{2}$}
is the so-called \emph{Mazur manifold} \cite{AkbKir:79,Akb:91} with
the \emph{Brieskorn-Sphere} $B(2,5,7)$ as boundary. Akbulut and its
coworkers \cite{Akbulut08,Akbulut09} discuss many examples of Akbulut
corks and the dependence of the smoothness structure on the cork.

For the following we need a short account of the proof of the h-cobordism
structure theorem. The interior of every h-cobordism can be divided
into pieces, called handle \cite{Mil:65}. A $k$-handle is the manifold
$D^{k}\times D^{5-k}$ which will be glued along the boundary $S^{k}\times D^{5-k}$.
The pairs of $0-/1-$ and $4-/5-$handles in a h-cobordism between
the two homeomorphic 4-manifolds $M_{0}$ and $M$ can be killed by
a general procedure (\cite{Mil:65}, \S8). Thus only the pairs of
$2-/3-$handles are left. Exactly these pairs are the difference between
the smooth h-cobordism and the topological h-cobordism. To eliminate
the $2-/3-$handles one has to embed a disk without self-intersections
into $M$ (Whitney trick). But that is mostly impossible in 4-dimensional
manifolds. Therefore Casson \cite{Cas:73} constructed by an infinite,
recursive process a special handle -- the Casson-handle $CH$ -- containing
the required disk without self-intersections. Freedman was able to
show topologically the existence of this disk and he constructs a
homeomorphism between every Casson handle $CH$ and the open 2-handle
$D^{2}\times\mathbb{R}^{2}$ \cite{Fre:82}. But $CH$ is in general
non-diffeomorphic to $D^{2}\times\mathbb{R}^{2}$ as shown later by
Gompf \cite{Gom:84,Gom:89}. 

Now we consider the smooth h-cobordism $W$ together with a neighborhood
$N$ of $2-/3-$handles. It is enough to assume a pair of handles
with two self-intersections (of opposite orientation) between the
2- and 3-Spheres at the boundary of the handle. Thus one can construct
an Akbulut cork $A$ in $M$ out of this data \cite{CuFrHsSt:97}.
The pair of $2-/3-$handles can be eliminated topologically by the
embedding of a Casson handle. Then as shown by Bizaca and Gompf \cite{BizGom:96}
the neighborhood $N$ of the handle pair as well the neighborhood
$N(A)$ of the embedded Akbulut cork consists of the cork $A$ and
the Casson handle $CH$. Especially the open neighborhood $N(A)$
of the Akbulut cork is an exotic $\mathbb{R}^{4}$. The situation
was analyzed in \cite{GomSti:1999}:
\begin{theorem}
\label{thm:fail-h-cobordism-exotic-R4}Let $W^{5}$ be a non-trivial
(smooth) h-cobordism between $M_{0}^{4}$ and $M^{4}$ (i.e. $W$
is not diffeomorphic to $M\times[0,1]$). Then there is an open sub-h-cobordism
$U^{5}$ that is homeomorphic to $\mathbb{R}^{4}\times[0,1]$ and
contains a compact contractable sub-h-cobordism $X$ (the cobordism
between the Akbulut corks, see above), such that both $W$ and $U$
are trivial cobordisms outside of $X$, i.e. one has the diffeomorphisms\[
W\setminus X=((W\cap M)\setminus X)\times[0,1]\quad\mbox{and}\quad U\setminus X=((U\cap M)\setminus X)\times[0,1]\]
(the latter can be chosen to be the restriction of the former). Furthermore
the open sets $U\cap M$ and $U\cap M_{0}$ are homeomorphic to $\mathbb{R}^{4}$
which are exotic $\mathbb{R}^{4}$ if $W$ is non-trivial.
\end{theorem}
Then one gets an exotic $\mathbb{R}^{4}$ which smoothly embeds automatically
in the 4-sphere, called a small exotic $\mathbb{R}^{4}$. Furthermore
we remark that the exoticness of the $\mathbb{R}^{4}$ is connected
with the non-trivial smooth h-cobordism $W^{5}$, i.e. the failure
of the smooth h-cobordism theorem implies the existence of small exotic
$\mathbb{R}^{4}$'s.

\subsection{Exotic $\mathbb{R}^{4}$ and Casson handles}

The theorem \ref{thm:fail-h-cobordism-exotic-R4} relates a non-trivial
h-cobordism between two compact, simple-connected, smooth 4-manifolds
to a small exotic $\mathbb{R}^{4}$. Using theorem \ref{thm:Akbulut-cork},
we can understand where the non-triviality of the h-cobordism comes
from: one of the Akbulut corks, say $A$, must be glued in by using
a non-trivial involution of the boundary $\partial A$. In the notation
above, there is a non-product h-cobordism $W$ between $M^{4}$ and
$M_{0}^{4}$ with a h-subcobordism $X$ between $A_{0}\subset M_{0}$
and $A\subset M$. There is an open neighborhood $U$ of the h-subcobordism
$X$ which is an open h-cobordism $U$ between the open neighborhoods
$N(A)\subset M,\: N(A_{0})\subset M_{0}$. Both neighborhoods are
homeomorphic to $\mathbb{R}^{4}$ but not diffeomorphic to the standard
$\mathbb{R}^{4}$ (as induced from the non-productness of the h-cobordism
$W$). This exotic $\mathbb{R}^{4}$ is the interior of the attachment
of a Casson handle $CH$ to the boundary $\partial A$ of the cork
$A$.

Now let us consider the basic construction of the Casson handle $CH$.
Let $M$ be a smooth, compact, simple-connected 4-manifold and $f:D^{2}\to M$
a (codimension-2) mapping. By using diffeomorphisms of $D^{2}$ and
$M$, one can deform the mapping $f$ to get an immersion (i.e. injective
differential) generically with only double points (i.e. $\#|f^{-1}(f(x))|=2$)
as singularities \cite{GolGui:73}. But to incorporate the generic
location of the disk, one is rather interesting in the mapping of
a 2-handle $D^{2}\times D^{2}$ induced by $f\times id:D^{2}\times D^{2}\to M$
from $f$. Then every double point (or self-intersection) of $f(D^{2})$
leads to self-plumbings of the 2-handle $D^{2}\times D^{2}$. A self-plumbing
is an identification of $D_{0}^{2}\times D^{2}$ with $D_{1}^{2}\times D^{2}$
where $D_{0}^{2},D_{1}^{2}\subset D^{2}$ are disjoint sub-disks of
the first factor disk%
\footnote{In complex coordinates the plumbing may be written as $(z,w)\mapsto(w,z)$
or $(z,w)\mapsto(\bar{w},\bar{z})$ creating either a positive or
negative (respectively) double point on the disk $D^{2}\times0$ (the
core).%
}. Consider the pair $(D^{2}\times D^{2},\partial D^{2}\times D^{2})$
and produce finitely many self-plumbings away from the attaching region
$\partial D^{2}\times D^{2}$ to get a kinky handle $(k,\partial^{-}k)$
where $\partial^{-}k$ denotes the attaching region of the kinky handle.
A kinky handle $(k,\partial^{-}k)$ is a one-stage tower $(T_{1},\partial^{-}T_{1})$
and an $(n+1)$-stage tower $(T_{n+1},\partial^{-}T_{n+1})$ is an
$n$-stage tower union kinky handles $\bigcup_{\ell=1}^{n}(T_{\ell},\partial^{-}T_{\ell})$
where two towers are attached along $\partial^{-}T_{\ell}$. Let $T_{n}^{-}$
be $(\mbox{interior}T_{n})\cup\partial^{-}T_{n}$ and the Casson handle
\[
CH=\bigcup_{\ell=0}T_{\ell}^{-}\]
is the union of towers (with direct limit topology induced from the
inclusions $T_{n}\hookrightarrow T_{n+1}$). A Casson handle is specified
up to (orientation preserving) diffeomorphism (of pairs) by a labeled
finitely-branching tree with base-point {*}, having all edge paths
infinitely extendable away from {*}. Each edge should be given a label
$+$ or $-$. Here is the construction: tree $\to CH$. Each vertex
corresponds to a kinky handle; the self-plumbing number of that kinky
handle equals the number of branches leaving the vertex. The sign
on each branch corresponds to the sign of the associated self plumbing.
The whole process generates a tree with infinite many levels. In principle,
every tree with a finite number of branches per level realizes a corresponding
Casson handle. The simplest non-trivial Casson handle is represented
by the tree $Tree_{+}$: each level has one branching point with positive
sign $+$. 

Given a labeled based tree $Q$, let us describe a subset $U_{Q}$
of $D^{2}\times D^{2}$. Now we will construct a $(U_{Q},\partial D^{2}\times D^{2})$
which is diffeomorphic to the Casson handle associated to $Q$. In
$D^{2}\times D^{2}$ embed a ramified Whitehead link with one Whitehead
link component for every edge labeled by $+$ leaving {*} and one
mirror image Whitehead link component for every edge labeled by $-$(minus)
leaving {*}. Corresponding to each first level node of $Q$ we have
already found a (normally framed) solid torus embedded in $D^{2}\times\partial D^{2}$.
In each of these solid tori embed a ramified Whitehead link, ramified
according to the number of $+$ and $-$ labeled branches leaving
that node. We can do that process for every level of $Q$. Let the
disjoint union of the (closed) solid tori in the $n$th family (one
solid torus for each branch at level $n$ in $Q$) be denoted by $X_{n}$.
$Q$ tells us how to construct an infinite chain of inclusions:\[
\ldots\subset X_{n+1}\subset X_{n}\subset X_{n-1}\subset\ldots\subset X_{1}\subset D^{2}\times\partial D^{2}\]
and we define the Whitehead decomposition $Wh_{Q}=\bigcap_{n=1}^{\infty}X_{n}$
of $Q$. $Wh_{Q}$ is the Whitehead continuum \cite{Whitehead35}
for the simplest unbranched tree. We define $U_{Q}$ to be\[
U_{Q}=D^{2}\times D^{2}\setminus(D^{2}\times\partial D^{2}\cup\mbox{closure}(Wh_{Q}))\]
alternatively one can also write\begin{equation}
U_{Q}=D^{2}\times D^{2}\setminus\mbox{cone}(Wh_{Q})\label{eq:UQ-diffeo-CH}\end{equation}
where $\mbox{cone}()$ is the cone of a space\[
cone(A)=A\times[0,1]/(x,0)\sim(x',0)\qquad\forall x,x'\in A\]
over the point $(0,0)\in D^{2}\times D^{2}$. As Freedman (see \cite{Fre:82}
Theorem 2.2) showed $U_{Q}$ is diffeomorphic to the Casson handle
$CH_{Q}$ given by the tree $Q$.

\subsection{The design of a Casson handle and its foliation\label{sub:design-of-CH}}

A Casson handle is represented by a labeled finitely-branching tree
$Q$ with base point $\star$, having all edge paths infinitely extendable
away from $\star$. Each edge should be given a label $+$ or $-$
and each vertex corresponds to a kinky handle where the self-plumbing
number of that kinky handle equals the number of branches leaving
the vertex. The open handle $D^{2}\times\mathbb{R}^{2}$ is represented
by the $\star$, i.e. there are no kinky handles. One of the cornerstones
of Freedmans proof of the homeomorphism between a Casson handle $CH$
and the open 2-handle $H=D^{2}\times\mathbb{R}^{2}$ are the reembedding
theorems. Then one foliates $CH$ and $H$ by copies of the frontier
$Fr(CH)$. The frontier of a set $K$ is defined by $Fr(K)=\mbox{closure}(\mbox{closure}(K)\setminus K)$.
As example we consider the interior $int(D^{2})$ of a disk and obtain
for the frontier $Fr(int(D^{2}))=\mbox{closure}(\mbox{closure}(int(D^{2}))\setminus int(D^{2}))=\partial D^{2}$,
i.e. the boundary of the disk $D^{2}$. Then Freedman (\cite{Fre:82}
p.398) constructs another labeled tree $S(Q)$ from the tree $Q$.
There is a base point from which a single edge (called {}``decimal
point'') emerges. The tree is binary: one edge enters and two edges
leaving a vertex (or every vertex is trivalent). The edges are named
by initial segments of infinite base 3-decimals representing numbers
in the standard {}``middle third'' Cantor set%
\footnote{This kind of Cantor set is given by the following construction: Start
with the unit Interval $S_{0}=[0,1]$ and remove from that set the
middle third and set $S_{1}=S_{0}\setminus(1/3,2/3)$ Continue in
this fashion, where $S_{n+1}=S_{n}\setminus\{\mbox{middle thirds of subintervals of \ensuremath{S_{n}}}\}$.
Then the Cantor set $C.s.$ is defined as $C.s.=\cap_{n}S_{n}$. With
other words, if we using a ternary system (a number system with base
3), then we can write the Cantor set as $C.s.=\{x:x=(0.a_{1}a_{2}a_{3}\ldots)\mbox{ where each \ensuremath{a_{i}=0}or \ensuremath{2}}\}$.%
} $CS\subset[0,1]$. Each edge $e$ of $S(Q)$ carries a label $\tau_{e}$
where $\tau_{e}$ is an ordered finite disjoint union of 5-stage towers
together with an ordered collection of standard loops generating the
fundamental group. There is three constraints on the labels which
leads to the correspondence between the $\pm$ labeled tree $Q$ and
the (associated) $\tau$-labeled tree $S(Q)$. One calls $S(Q)$ the
design.

Two words are in order for the design $S(Q)$: first, every sequence
of $0$'s and $2$'s is one path in $S(Q)$ representing one embedded
Casson handle $CH_{Q_{1}}\subset CH_{Q}$ where both trees are related
like $Q\subset Q_{1}$. For example, the Casson handle corresponding
to $.020202...$ is obtained as the union of the 5-stage towers $T^{0}\cup T^{02}\cup T^{020}\cup T^{0202}\cup T^{02020}\cup T^{020202}\cup...$.
For later usage we identify the sequence $.00000...$ with the Tree
$Tree_{+}$. Secondly, there are gaps, i.e. we have only a Cantor
set of Casson handles not a continuum. For instance a gap is lying
between the paths $.022222\ldots$ and $.20000\ldots$ In the proof
of Freedman, the gaps are shrunk to a point and one gets the desired
homeomorphism. Here we will use this structure to produce a foliation
of the design. Every path in $S(Q)$ is represented by one sequence
over the alphabet $\{\text{0,2\}}$. Every gap is a sequence containing
at least one $1$ (so for instance $.1222...$ or $.012222...)$.
There is now a natural order structure given by the sequence (for
instance $.022222...<.12222...<.22222...$). The leaves are the corresponding
gaps or Casson handles (represented by the union 5-stage towers ending
with $T^{02222...},\, T^{12222..}\mbox{ or }T^{22222...}$). The tree
structure of the design $S(Q)$ should be also reflected in the foliation
to represent every path in $S(Q)$ as a union of 5-stage towers. By
the reembedding theorems, the 5-stage towers can be embedded into
each other. Then we obtain two foliations of the (topological) open
2-handle $D^{2}\times\mathbb{R}^{2}$: a codimension-1 foliation along
one $\mathbb{R}-$axis labeled by the sequences (for instance $.022222...<.12222...<.22222...$)
and a second codimension-1 foliation along the radius of the disk
$D^{2}$ induced by inclusion of the 5-stage towers (for instance
$T^{0}\supset T^{02}\supset T^{020}\supset...$). Especially the exploration
of a Casson handle by using the design is given by its frontier, in
this case, minus the attaching region. In case of a usual tower we
get the frontier $S^{1}\times D^{2}/Wh_{\gamma}$ with $\gamma\in S(Q)$.
The gaps have a similar structure. Then the foliation of the Casson
handle (induced from the design) is given by the leaves $S^{1}\times D^{1}$
over the disk $D^{2}$ in the Casson handle, i.e. the disk $D^{2}$
is foliated by parallel lines (see Fig. \ref{fig:foliation-of-the-disk}).
\begin{figure}
\includegraphics[scale=0.5]{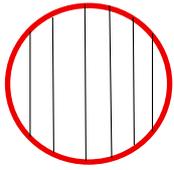}

\caption{foliation of the disk $D^{2}$ in the design $S(Q)$ \label{fig:foliation-of-the-disk}}

\end{figure}
 So, every Casson handle with a given tree $Q$ has a codimension-one
foliation given by its design.

This foliation can be also understood as a foliated cobordism. For
that purpose we consider the foliation as part of a foliation of the
2-sphere (see Fig. \ref{fig:foliation-of-S2}). %
\begin{figure}
\includegraphics[scale=0.5]{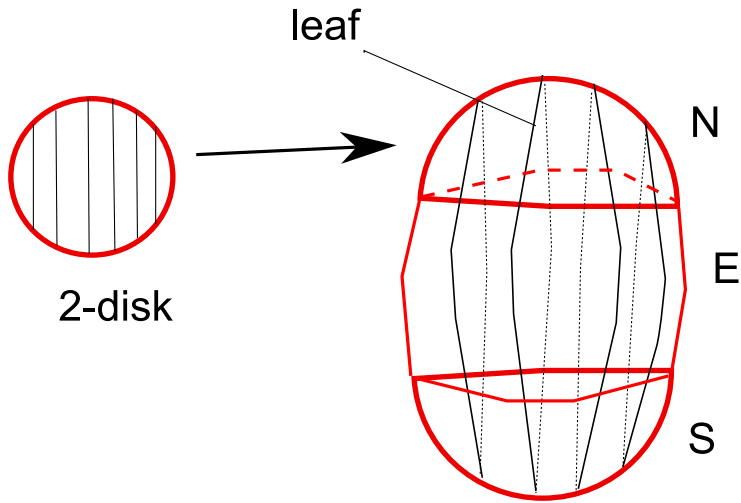}

\caption{foliation of the 2-sphere as foliated cobordism $E$ of the two disks
$N,S$ \label{fig:foliation-of-S2}}

\end{figure}
 The 2-sphere is decomposed by $S^{2}=N\cup E\cup S$, two pole regions
$N,S$ ($N,S=D^{2})$ and an equator region $E=S^{1}\times D^{1}$.
The foliation of the disk as in Fig. \ref{fig:foliation-of-the-disk}
can be used to foliate $N$ and $S$. Both foliations can be connected
by the leaves $S^{1}$ which are the longitudes. Then one obtains
a foliated cobordism between $N$ and $S$ given by the obvious foliation
of the equator region $E$ (a cylinder).

\subsection{Capped gropes and its design\label{sub:Capped-gropes}}

In this subsection we discuss a possible generalization of Casson
handles. The modern way to the classification of 4-manifolds used
{}``capped gropes'', a mixed variant of Casson handle and grope
(chapters 1 to 4 in \cite{FreQui:90}). We do not want to complicate
the situation more than needed. But for later developments we have
to discuss some part of the theory but we remark that all results
can be easily generalized to capped gropes as well.

A grope is a special pair (2-complex,circle), where the circle is
referred to as the boundary of the grope. There is an anomalous case
when the depth is $1$: the unique grope of depth 1 is the pair (circle,circle).
A grope of depth 2 is a punctured surface with the boundary circle
specified (see Fig. \ref{fig:Example-of-grope}). %
\begin{figure}
\includegraphics[scale=0.25]{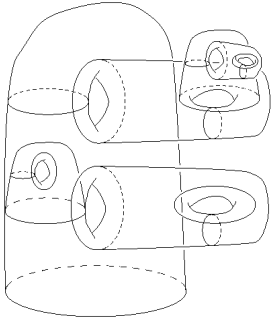}

\caption{Example of a grope with symplectic basis as curves around the holes\label{fig:Example-of-grope}}

\end{figure}
To form a grope $G$ of depth $n$, take a punctured surface, $F$,
and prescribe a symplectic basis $\left\{ \alpha_{i},\beta_{j}\right\} $.
That is, $\alpha_{i}$ and $\beta_{j}$ are embedded curves in $F$
which represent a basis of $H_{1}(F)$ such that the only intersections
among the $\alpha_{i}$ and $\beta_{j}$ occur when $\alpha_{i}$
and $\beta_{j}$ meet in a single point $\alpha_{i}\cdot\beta_{j}=1$.
Now glue gropes of depth $<n$ along their boundary circles to each
$\alpha_{i}$ and $\beta_{j}$ with at least one such added grope
being of depth $n-1$. (Note that we are allowing any added grope
to be of depth 1, in which case we are not really adding a grope.)
The surface $F\subset G$ is called the bottom stage of the grope
and its boundary is the boundary of the grope. The tips of the grope
are those symplectic basis elements of the various punctured surfaces
of the grope which do not have gropes of depth $>1$ attached to them. 
\begin{definition}
A capped grope is a grope with disks (the caps) attached to all its
tips. The grope without the caps is sometimes called the body of the
capped grope. 
\end{definition}
The capped grope (as cope) was firstly described by Freedman in 1983\cite{Fre:83}.
The caps are only immersed disks like in case of the Casson handle
to make the grope simple-connected. The great advantage is the simpler
frontier, instead of $S^{1}\times D^{2}/Wh_{\gamma}$ (see subsection
\ref{sub:design-of-CH}) one has solid tori $S^{1}\times D^{2}$ as
frontier of the capped grope (as shown in \cite{AncelStarbird1989}).
The corresponding design (and its parametrization) can be described
similar to the Casson handle by sequences containing $0$ and $2$
(see section 4.5 in \cite{FreQui:90}). There are also gaps (described
by a $1$ in the sequence) who look like $S^{1}\times D^{3}$.

\subsection{The radial family of uncountable-many small exotic $\mathbb{R}^{4}$}

Given a small exotic $\mathbb{R}^{4}$ $R$ induced from the non-product
h-cobordism $W$ between $M$ and $M_{0}$ with Akbulut corks $A\subset M$
and $A_{0}\subset M_{0}$, respectively. Let $K\subset\mathbb{R}^{4}$
be a compact subset. Bizaca and Gompf \cite{BizGom:96} constructed
the small exotic $\mathbb{R}_{1}^{4}$ by using the simplest tree
$Tree_{+}$. Bizaca \cite{Biz:94,Bizaca1995} showed that the Casson
handle generated by $Tree_{+}$ is an exotic Casson handle. Using
Theorem 3.2 of \cite{DeMichFreedman1992}, there is a topological
radius function $\rho:\mathbb{R}_{1}^{4}\to[0,+\infty)$ (polar coordinates)
so that $\mathbb{R}_{t}^{4}=\rho^{-1}([0,r))$ with $t=1-\frac{1}{r}$.
Then $K\subset\mathbb{R}_{0}^{4}$ and $\mathbb{R}_{t}^{4}$ is also
a small exotic $\mathbb{R}^{4}$ for $t$ belonging to a Cantor set
$CS\subset[0,1]$. Especially two exotic $\mathbb{R}_{s}^{4}$ and
$\mathbb{R}_{t}^{4}$ are non-diffeomorphic for $s<t$ except for
countable many pairs. In \cite{DeMichFreedman1992} it was claimed
that there is a smoothly embedded homology 4-disk $A$. The boundary
$\partial A$ is a homology 3-sphere with a non-trivial representation
of its fundamental group into $SO(3)$ (so $\partial A$ cannot be
diffeomorphic to a 3-sphere). According to Theorem \ref{thm:fail-h-cobordism-exotic-R4}
this homology 4-disk must be identified with the Akbulut cork of the
non-trivial h-cobordism. The cork $A$ is contractable and can be
(at least) build by one 1-handle and one 2-handle (case of a Mazur
manifold). Given a radial family $\mathbb{R}_{t}^{4}$ with radius
$r=\frac{1}{1-t}$ so that $t=1-\frac{1}{r}\subset CS\subset[0,1]$.
Suppose there is a diffeomorphism\[
(d,id_{K}):(\mathbb{R}_{s}^{4},K)\to(\mathbb{R}_{t}^{4},K)\qquad s\not=t\in CS\]
fixing the compact subset $K$. Then this map $d$ induces end-periodic
manifolds%
\footnote{We ignore the inclusion for simplicity.%
} $M\setminus(\bigcap_{i=0}^{\infty}d^{i}(\mathbb{R}_{s}^{4}))$ and
$M_{0}\setminus(\bigcap_{i=0}^{\infty}d^{i}(\mathbb{R}_{s}^{4}))$
which must be smoothable contradicting a theorem of Taubes \cite{Tau:87}.
Therefore $\mathbb{R}_{s}^{4}$ and $\mathbb{R}_{t}^{4}$ are non-diffeomorphic
for $t\not=s$ (except for countable many possibilities).

\subsection{Exotic $\mathbb{R}^{4}$ and codimension-1 foliations}

In this subsection we will construct a codimension-one foliation on
the boundary $\partial A$ of the cork with non-trivial Godbillon-Vey
invariant. The strategy of the proof goes like this: we use the foliation
of the design of the Casson handle (see subsection \ref{sub:design-of-CH})
for the radial family $\mathbb{R}_{t}^{4}$ to induce a foliated cobordism
$\partial A\times[0,1]$. The restriction to its boundary gives cobordant
codimension-1 foliations of $\partial A$ with non-trivial Godbillon-Vey
invariant $r^{2}=\frac{1}{(1-t)^{2}}$. In the subsection \ref{sub:Non-cobordant-foliations-S3}
we described a foliation of the 3-sphere unique up to foliated cobordism
for every given value of the Godbillon-Vey invariant. By theorem \ref{thm:foliation-3MF},
we get a corresponding foliation on $\partial A$ (with the same Godbillon-Vey
invariant). Finally we obtain:
\begin{theorem}
\label{thm:codim-1-foli-radial-fam}Given a radial family $R_{t}$
of small exotic $\mathbb{R}_{t}^{4}$ with radius $r$ and $t=1-\frac{1}{r}\subset CS\subset[0,1]$
induced from the non-product h-cobordism $W$ between $M$ and $M_{0}$
with Akbulut cork $A\subset M$ and $A\subset M_{0}$, respectively.
The radial family $R_{t}$ determines a family of codimension-one
foliations of $\partial A$ with Godbillon-Vey invariant $r^{2}$.
Furthermore given two exotic spaces $R_{t}$ and $R_{s}$, homeomorphic
but non-diffeomorphic to each other (and so $t\not=s$) then the two
corresponding codimension-one foliation of $\partial A$ are non-cobordant
to each other.\end{theorem}
\begin{proof}
The proof can be found in \cite{AsselmeyerKrol2009} and in the appendix
A. $\square$
\end{proof}
In the theorem \ref{thm:codim-1-foli-radial-fam} above we constructed
a relation between codimension-1 foliations on $\Sigma=\partial A$
and the radius for the radial family of small exotic $\mathbb{R}^{4}$.
By using theorem \ref{thm:foliation-3MF} we can trace back the foliation
on $\Sigma$ by a foliation on the 3-sphere $S^{3}$. This situation
can be seen differently by using the diffeomorphism $\Sigma=\Sigma\#S^{3}$.
Then, the foliation on $S^{3}$ induces a foliation on $\Sigma$ at
least partially. Thus, we have a 3-sphere $S^{3}$ lying at the boundary
$\partial A=\Sigma$ of the Akbulut cork $A$ inducing a codimension-1
foliation on $\Sigma$. Then by theorem \ref{thm:codim-1-foli-radial-fam}:
\begin{corollary}
Any class in $H^{3}(S^{3},\mathbb{R})$ induces a small exotic $\mathbb{R}^{4}$
where $S^{3}$ lies at the boundary $\Sigma=\partial A$ of the cork
$A$.
\end{corollary}

\subsection{Integer Godbillon-Vey invariants and flat bundles\label{sub:Integer-Godbillon-Vey-invariants}}

Clearly the integer classes $H^{3}(S^{3},\mathbb{Z})\subset H^{3}(S^{3},\mathbb{R})$
are a subset of the full set and one can use the construction above
to get the foliation. Especially the polygon $P$ must be formed by
segments with angles $\alpha_{k}$ of integer value with respect to
$\pi$ to get an integer value for the volume $Area(P)=(k-2)\pi-\sum_{k}\alpha_{k}$
up to a $\pi-$factor. Using the work of Goldman and Brooks \cite{BrooksGoldman1984},
one can construct a foliation admitting an integer Godbillon-Vey invariant.
The corresponding foliation is induced by the unit tangent $T_{1}\mathbb{H}^{2}$
or by the action of the M�bius group $PSL(2,\mathbb{R})=SL(2,\mathbb{R})/\mathbb{Z}_{2}$
(Remark: $PSL(2,\mathbb{R})$ acts transitively on $T_{1}\mathbb{H}^{2}$
and so we can identify both spaces). The unit tangent bundle $T_{1}\mathbb{H}^{2}=PSL(2,\mathbb{R})$
is a circle bundle over $\mathbb{H}^{2}$ and we can construct the
universal cover, a real line bundle over $\mathbb{H}^{2}$, denoted
by $\widetilde{SL}(2,\mathbb{R})$. In subsection \ref{sub:Non-cobordant-foliations-S3}
we described Thurstons construction of a codimension-1 foliation $\mathcal{F}$.
In an intermediate step one has the manifold $M=(S^{2}\setminus\left\{ \mbox{\mbox{k} punctures}\right\} )\times S^{1}$
(with a foliation $\mathcal{F}$). This foliation $\mathcal{F}$ is
defined by a one-form $\omega$ together with two other 1-forms $\theta,\eta$
with\begin{equation}
d\omega=\theta\wedge\omega,\quad d\theta=\omega\wedge\eta,\quad d\eta=\eta\wedge\theta\label{eq:sl2-relations}\end{equation}
and Godbillon-Vey invariant $GV(\mathcal{F})=\theta\wedge d\theta=\omega\wedge\eta\wedge\theta$.
Now we show that the Godbillon-Vey invariant of this foliation $\mathcal{F}$
is an integer 3-form:
\begin{lemma}
\label{lem:integer-GV}Given a manifold $M$ with non-trivial fundamental
group $\pi_{1}(M)$ with foliation $\mathcal{F}$ defined by the 1-form
$\omega$ together with two 1-forms $\theta,\eta$ fulfilling the
relations (\ref{eq:sl2-relations}). If $M$ can be written as a flat
$PSL(2,\mathbb{R})-$bundle over a manifold $N$ with fiber $S^{1}$
and $\pi_{1}(N)\not=0$. Then the pairing of the Godbillon-Vey invariant
with the fundamental class $[M]\in H_{3}(M)$ is given by\begin{equation}
\langle GV(\mathcal{F}),[M]\rangle=\intop_{M}GV(\mathcal{F})=(4\pi)^{2}\cdot\chi(N)\label{eq:integer-GV}\end{equation}
with the Euler characteristics $\chi(M)$ of $N$. Up to a normalization
constant one obtains an integer value.\end{lemma}
\begin{proof}
The proof can be found in \cite{AsselmeyerKrol2009}. $\square$

Using this lemma we are able to obtain the special foliation (a la
Thurston) of the $S^{3}$ with integer Godbillon-Vey invariant.\end{proof}
\begin{theorem}
\label{thm:integer-GV-flat-bundles}Every $PSL(2,\mathbb{R})$ flat
bundle over $M=(S^{2}\setminus\left\{ \mbox{\mbox{k} punctures}\right\} )\times S^{1}$
defines a codimension-1 foliation of $M$ by the horizontal distribution
of the flat connection so that its (normalized) Godbillon-Vey invariant
is an integer given by\begin{equation}
\frac{1}{(4\pi)^{2}}\langle GV(\mathcal{F}),[M]\rangle=\pm\chi(N)=\pm(2-k)\quad.\label{eq:integer-GV-S3}\end{equation}
This foliation can be extended to the whole 3-sphere $S^{3}$ defining
an integer class in $H^{3}(S^{3},\mathbb{Z})$.\end{theorem}
\begin{proof}
The proof can be found in \cite{AsselmeyerKrol2009}. The sign of
the integral depends on the orientation of the manifold $M$. $\square$

It is an important consequence of the work \cite{BrooksGoldman1984}
that the foliation $\mathcal{F}$ (and its induced counterpart for
the 3-sphere $S^{3}$) is rigid, i.e. a disturbance (or continuous
variation) does not change the Godbillon-Vey invariant. 
\end{proof}

\subsection{From exotic smoothness to operator algebras\label{sub:From-exotic-smoothness}}

In subsection \ref{sub:smooth-holonomy-groupoid} we constructed (following
Connes \cite{Connes94}) the smooth holonomy groupoid of a foliation
$F$ and its operator algebra $C_{r}^{*}(M,F)$. The correspondence
between a foliation and the operator algebra (as well as the von Neumann
algebra) is visualized by table \ref{tab:relation-foliation-operator}.
\begin{table}
\begin{tabular}{|c|c|}
\hline 
Foliation & Operator algebra\tabularnewline
\hline
\hline 
leaf & operator\tabularnewline
\hline 
closed curve transversal to foliation & projector (idempotent operator)\tabularnewline
\hline 
holonomy & linear functional (state)\tabularnewline
\hline 
local chart & center of algebra\tabularnewline
\hline
\end{tabular}

\caption{relation between foliation and operator algebra\label{tab:relation-foliation-operator}}

\end{table}
 As extract of our previous paper \cite{AsselmeyerKrol2009}, we obtained
a relation between exotic $\mathbb{R}^{4}$'s and codimension-1 foliations
of the 3-sphere $S^{3}$. For a codimension-1 foliation there is the
Godbillon-Vey invariant \cite{GodVey:71} as element of $H^{3}(M,\mathbb{R})$.
Hurder and Katok \cite{HurKat:84} showed that the $C^{*}$algebra
of a foliation with non-trivial Godbillon-Vey invariant contains a
factor $I\! I\! I$ subalgebra (by the Anosov-like foliation). Using
Tomita-Takesaki-theory, one has a continuous decomposition (as crossed
product) of any factor $I\! I\! I$ algebra $M$ into a factor $I\! I_{\infty}$
algebra $N$ together with a one-parameter group%
\footnote{The group $\mathbb{R}_{+}^{*}$ is the group of positive real numbers
with multiplication as group operation also known as Pontrjagin dual.%
} $\left(\theta_{\lambda}\right)_{\lambda\in\mathbb{R}_{+}^{*}}$ of
automorphisms $\theta_{\lambda}\in Aut(N)$ of $N$, i.e. one obtains

\[
M=N\rtimes_{\theta}\mathbb{R}_{+}^{*}\quad.\]
But that means, there is a foliation induced from the foliation of
the $S^{3}$ producing this $I\! I_{\infty}$ factor. As we saw in
subsection \ref{sub:factor-III-case} one has a codimension-1 foliation
$F$ as part of the foliation of the $S^{3}$ whose von Neumann algebra
is the hyperfinite factor $I\! I\! I_{1}$. Connes \cite{Connes94}
(in section I.4 page 57ff) constructed the foliation $F'$ canonically
associated to $F$ having the factor $I\! I_{\infty}$ as von Neumann
algebra. In our case it is the horocycle flow: Let $P$ the polygon
on the hyperbolic space $\mathbb{H}^{2}$ determining the foliation
of the $S^{3}$ (see subsection ). $P$ is equipped with the hyperbolic
metric $2|dz|/(1-|z|^{2})$ together with the collection $T_{1}P$
of unit tangent vectors to $P$. A horocycle in $P$ is a circle contained
in $P$ which touches $\partial P$ at one point (see Fig. \ref{fig:horocycle-fig}).
\begin{figure}
\includegraphics[scale=0.25]{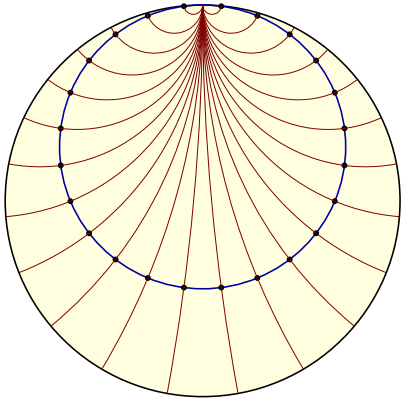}

\caption{horocycle, a curve whose normals all converge asymptotically\label{fig:horocycle-fig}}

\end{figure}
 Then the horocycle flow $T_{1}P\to T_{1}P$ is the flow moving an
unit tangent vector along a horocycle (in positive direction at unit
speed). As above the polygon $P$ determines a surface $S$ of genus
$g>1$ with abelian torsion-less fundamental group $\pi_{1}(S)$ so
that the homomorphism $\pi_{1}(S)\to\mathbb{R}$ determines an unique
(ergodic invariant) Radon measure. Finally the horocycle flow determines
a factor $I\! I_{\infty}$ foliation associated to the factor $I\! I\! I_{1}$
foliation. We remark for later usage that this foliation is determined
by a set of closed curves (the horocycles). Using results of previous
papers and subsections above, we have the following picture:
\begin{enumerate}
\item Every small exotic $\mathbb{R}^{4}$ is determined by a codimension-1
foliation (unique up to cobordisms) of some homology 3-sphere $\Sigma$
(as boundary $\partial A=\Sigma$ of a contractable submanifold $A\subset\mathbb{R}^{4}$,
the Akbulut cork). (see Theorem \ref{thm:codim-1-foli-radial-fam})
\item This codimension-1 foliation on $\Sigma$ determines via surgery along
a link uniquely a codimension-1 foliation on the 3-sphere and vice
verse. (see Theorem \ref{thm:foliation-3MF})
\item This codimension-1 foliation $(S^{3},F)$ on $S^{3}$ has a leaf space
which is determined by the von Neumann algebra $W(S^{3},F)$ associated
to the foliation. (see Connes \cite{Connes1984})
\item The von Neumann algebra $W(S^{3},F)$ contains a hyperfinite factor
$I\! I\! I_{1}$ algebra as well as a factor $I_{\infty}$ algebra
coming from the Reeb foliations. (see Hurder and Katok \cite{HurKat:84})
\end{enumerate}
Thus by this procedure we get a noncommutative algebra from an exotic
$\mathbb{R}^{4}$. The relation to the quantum theory will be discussed
now. We remark that we have already a quantum theory represented by
the von Neumann algebra $W(S^{3},F)$. Thus we are in the strange
situation to construct a (classical) Poisson algebra together with
a quantization to get an algebra which we already have.

\section{The connection between exotic smoothness and quantization\label{sec:Quantization}}

In this section we describe a deep relation between quantization and
the codimension-1 foliation of the $S^{3}$ determining the smoothness
structure on a small exotic $\mathbb{R}^{4}$. Her and in the following
\emph{we will identify the leaf space with its operator algebra}.

\subsection{Idempotent operators, closed curves in surfaces and knot cobordisms\label{sub:Idempotent-operators-knot-cobordism}}

In subsection \ref{sub:smooth-holonomy-groupoid}, an idempotent was
constructed in the $C^{*}$ algebra of the foliation and geometrically
interpreted as closed curve transversal to the foliation. Such a curve
meets every leaf in a finite number of points. Furthermore the foliation
on the 3-sphere $S^{3}$ is determined up to foliated cobordisms (see
Theorem \ref{thm:codim-1-foli-radial-fam}), i.e. a 4-space which
looks like $S^{3}\times[0,1]$. Then we have a cobordism of two curves
which looks like a thickened curve $S^{1}\times[0,1]$. The foliation
of the 3-sphere is determined by a polygon $P$ (see subsection \ref{sub:Non-cobordant-foliations-S3})
laminated by curves starting and ending at the boundary of $P$ (see
Fig. \ref{fig:horocycle-fig}). Without loss of generality we can
assume that $P$ consists of an even number of vertices, say $2k$.
By the uniformization theorem of surfaces, there is an unique surface
$S$ of genus $g>1$ with $g=[k/2]$ (with one boundary component
for $k$ odd) represented by $P$. The closed curves at $S$ transversal
to the foliation are represented by lines perpendicular to the leafs
in the foliation of $P$ (see Fig. \ref{fig:closed-curves-in-P}).%
\begin{figure}
\includegraphics[scale=0.2]{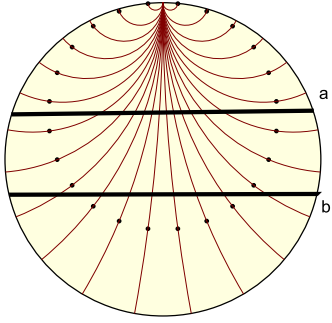}

\caption{$a,b$ are examples of closed curves in $P$ where the polygon is
simplified by a circle\label{fig:closed-curves-in-P}}

\end{figure}
 In the process from $P$ to the surface $S$ (via the identification
of sides of $P$) these lines in $P$ closes to curves in $S$. How
do these curves look like? Usually a closed curve in the foliation
of the 3-sphere is given by an embedding $S^{1}\to S^{3}$ where the
normal of this curve is in the direction of the leaf. This embedding
is also known as a knot. In the process from $P$ to $S$, we project
this knot to the surface and get a closed curve with singularities
(instead of crossings). Then this closed curve is represented by a
line (in $P$) perpendicular to the flow lines of the foliation. Now
we will state the following theorem:
\begin{theorem}
\label{thm:idempotent-operators-as-closed-curves}Given a codimension-1
foliation of the 3-sphere represented by a polygon $P$ with $4k$
vertices. The corresponding Anosov foliation of the unit tangent bundle
$T_{1}S$ of a surface $S$ with genus $g=k$ admits idempotent operators
in the leaf space given by closed curves in $S$ with self-intersections.
Two closed curves are equivalent if there is an isotopy between both
curves.\end{theorem}
\begin{proof}
The proof is a simple combination of results presented above. By definition
(see \ref{sub:smooth-holonomy-groupoid}), one has an idempotent operator
in the leaf space of the foliation which is a closed curve transversal
to the foliation. These curves are an embedding $S^{1}\to S^{3}$
known as knot. A deformation between two curves is an isotopy. Then
both curves are equivalent. $\square$

Thus we have to consider closed curves in surfaces. 
\end{proof}

\subsection{The observable algebra and Poisson structure\label{sub:The-observable-algebra}}

In this section we will describe the formal structure of a classical
theory coming from the algebra of observables using the concept of
a Poisson algebra. In quantum theory, an observable is represented
by a hermitean operator having the spectral decomposition via projectors
or idempotent operators. The coefficient of the projector is the eigenvalue
of the observable or one possible result of a measurement. At least
one of these projectors represent (via the GNS representation) a quasi-classical
state. Thus to construct the substitute of a classical observable
algebra with Poisson algebra structure we have to concentrate on the
idempotents in the $C^{*}$ algebra. Now we will see that the set
of closed curves on a surface has the structure of a Poisson algebra.
Let us start with the definition of a Poisson algebra. 
\begin{definition}
Let $P$ be a commutative algebra with unit over $\mathbb{R}$ or
$\mathbb{C}$. A \emph{Poisson bracket} on $P$ is a bilinearform
$\left\{ \:,\:\right\} :P\otimes P\to P$ fulfilling the following
3 conditions:

anti-symmetry $\left\{ a,b\right\} =-\left\{ b,a\right\} $

Jacobi identity $\left\{ a,\left\{ b,c\right\} \right\} +\left\{ c,\left\{ a,b\right\} \right\} +\left\{ b,\left\{ c,a\right\} \right\} =0$

derivation $\left\{ ab,c\right\} =a\left\{ b,c\right\} +b\left\{ a,c\right\} $.\\
Then a \emph{Poisson algebra} is the algebra $(P,\{\,,\,\})$.
\end{definition}
Now we consider a surface $S$ together with a closed curve $\gamma$.
Additionally we have a Lie group $G$ given by the isometry group.
The closed curve is one element of the fundamental group $\pi_{1}(S)$.
From the theory of surfaces we know that $\pi_{1}(S)$ is a free abelian
group. Denote by $Z$ the free $\mathbb{K}$-module ($\mathbb{K}$
a ring with unit) with the basis $\pi_{1}(S)$, i.e. $Z$ is a freely
generated $\mathbb{K}$-modul. Recall Goldman's definition of the
Lie bracket in $Z$ (see \cite{Goldman1984}). For a loop $\gamma:S^{1}\to S$
we denote its class in $\pi_{1}(S)$ by $\left\langle \gamma\right\rangle $.
Let $\alpha,\beta$ be two loops on $S$ lying in general position.
Denote the (finite) set $\alpha(S^{1})\cap\beta(S^{1})$ by $\alpha\#\beta$.
For $q\in\alpha\#\beta$ denote by $\epsilon(q;\alpha,\beta)=\pm1$
the intersection index of $\alpha$ and $\beta$ in $q$. Denote by
$\alpha_{q}\beta_{q}$ the product of the loops $\alpha,\beta$ based
in $q$. Up to homotopy the loop $(\alpha_{q}\beta_{q})(S^{1})$ is
obtained from $\alpha(S^{1})\cup\beta(S^{1})$ by the orientation
preserving smoothing of the crossing in the point $q$. Set \begin{equation}
[\left\langle \alpha\right\rangle ,\left\langle \beta\right\rangle ]=\sum_{q\in\alpha\#\beta}\epsilon(q;\alpha,\beta)(\alpha_{q}\beta_{q})\quad.\label{eq:Lie-bracket-loops}\end{equation}
According to Goldman \cite{Goldman1984}, Theorem 5.2, the bilinear
pairing $[\,,\,]:Z\times Z\to Z$ given by (\ref{eq:Lie-bracket-loops})
on the generators is well defined and makes $Z$ to a Lie algebra.
The algebra $Sym(Z)$ of symmetric tensors is then a Poisson algebra
(see Turaev \cite{Turaev1991}).

The whole approach seems natural for the construction of the Lie algebra
$Z$ but the introduction of the Poisson structure is an artificial
act. From the physical point of view, the Poisson structure is not
the essential part of classical mechanics. More important is the algebra
of observables, i.e. functions over the configuration space forming
the Poisson algebra. Thus we will look for the algebra of observables
in our case. For that purpose, we will look at geometries over the
surface. By the uniformization theorem of surfaces, there is three
types of geometrical models: spherical $S^{2}$, Euclidean $\mathbb{E}^{2}$
and hyperbolic $\mathbb{H}^{2}$. Let $\mathcal{M}$ be one of these
models having the isometry group $Isom(\mathcal{M})$. Consider a
subgroup $H\subset Isom(\mathcal{M})$ of the isometry group acting
freely on the model $\mathcal{M}$ forming the factor space $\mathcal{M}/H$.
Then one obtains the usual (closed) surfaces $S^{2}$, $\mathbb{R}P^{2}$,
$T^{2}$ and its connected sums like the surface of genus $g$ ($g>1$).
For the following construction we need a group $G$ containing the
isometry groups of the three models. Furthermore the surface $S$
is part of a 3-manifold and for later use we have to demand that $G$
has to be also a isometry group of 3-manifolds. According to Thurston
\cite{Thu:97} there are 8 geometric models in dimension 3 and the
largest isometry group is the hyperbolic group $PSL(2,\mathbb{C})$
isomorphic to the Lorentz group $SO(3,1).$ It is known that every
representation of $PSL(2,\mathbb{C})$ can be lifted to the spin group
$SL(2,\mathbb{C})$. Thus the group $G$ fulfilling all conditions
is identified with $SL(2,\mathbb{C})$. This choice fits very well
with the 4-dimensional picture.

Now we introduce a principal $G$ bundle on $S$, representing a geometry
on the surface. This bundle is induced from a $G$ bundle over $S\times[0,1]$
having always a flat connection. Alternatively one can consider a
homomorphism $\pi_{1}(S)\to G$ represented as holonomy functional\begin{equation}
hol(\omega,\gamma)=\mathcal{P}\exp\left(\int\limits _{\gamma}\omega\right)\in G\label{eq:holonomy-definition}\end{equation}
with the path ordering operator $\mathcal{P}$ and $\omega$ as flat
connection (i.e. inducing a flat curvature $\Omega=d\omega+\omega\wedge\omega=0$).
This functional is unique up to conjugation induced by a gauge transformation
of the connection. Thus we have to consider the conjugation classes
of maps\[
hol:\pi_{1}(S)\to G\]
forming the space $X(S,G)$ of gauge-invariant flat connections of
principal $G$ bundles over $S$. Now (see \cite{Skovborg2006}) we
can start with the construction of the Poisson structure on $X(S,G).$
The construction based on the Cartan form as the unique bilinearform
of a Lie algebra. As discussed above we will use the Lie group $G=SL(2,\mathbb{C})$
but the whole procedure works for every other group too. Now we consider
the standard basis\[
X=\left(\begin{array}{cc}
0 & 1\\
0 & 0\end{array}\right)\quad,\qquad H=\left(\begin{array}{cc}
1 & 0\\
0 & -1\end{array}\right)\quad,\qquad Y=\left(\begin{array}{cc}
0 & 0\\
1 & 0\end{array}\right)\]
of the Lie algebra $sl(2,\mathbb{C})$ with $[X,Y]=H,\,[H,X]=2X,\,[H,Y]=-2Y$.
Furthermore there is the bilinearform $B:sl_{2}\otimes sl_{2}\to\mathbb{C}$
written in the standard basis as \[
\left(\begin{array}{ccc}
0 & 0 & -1\\
0 & -2 & 0\\
-1 & 0 & 0\end{array}\right)\]
Now we consider the holomorphic function $f:SL(2,\mathbb{C})\to\mathbb{C}$
and define the gradient $\delta_{f}(A)$ along $f$ at the point $A$
as $\delta_{f}(A)=Z$ with $B(Z,W)=df_{A}(W)$ and \[
df_{A}(W)=\left.\frac{d}{dt}f(A\cdot\exp(tW))\right|_{t=0}\quad.\]
The calculation of the gradient $\delta_{tr}$ for the trace $tr$
along a matrix \[
A=\left(\begin{array}{cc}
a_{11} & a_{12}\\
a_{21} & a_{22}\end{array}\right)\]
 is given by\[
\delta_{tr}(A)=-a_{21}Y-a_{12}X-\frac{1}{2}(a_{11}-a_{22})H\quad.\]
Given a representation $\rho\in X(S,SL(2,\mathbb{C}))$ of the fundamental
group and an invariant function $f:SL(2,\mathbb{C})\to\mathbb{R}$
extendable to $X(S,SL(2,\mathbb{C}))$. Then we consider two conjugacy
classes $\gamma,\eta\in\pi_{1}(S)$ represented by two transversal
intersecting loops $P,Q$ and define the function $f_{\gamma}:X(S,SL(2,\mathbb{C})\to\mathbb{C}$
by $f_{\gamma}(\rho)=f(\rho(\gamma))$. Let $x\in P\cap Q$ be the
intersection point of the loops $P,Q$ and $c_{x}$ a path between
the point $x$ and the fixed base point in $\pi_{1}(S)$. The we define
$\gamma_{x}=c_{x}\gamma c_{x}^{-1}$ and $\eta_{x}=c_{x}\eta c_{x}^{-1}$.
Finally we get the Poisson bracket\[
\left\{ f_{\gamma},f'_{\eta}\right\} =\sum_{x\in P\cap Q}sign(x)\: B(\delta_{f}(\rho(\gamma_{x})),\delta_{f'}(\rho(\eta_{x})))\quad,\]
where $sign(x)$ is the sign of the intersection point $x$. Thus,
\begin{theorem}
The space $X(S,SL(2,\mathbb{C}))$ has a natural Poisson structure
(induced by the bilinear form (\ref{eq:Lie-bracket-loops}) on the
group) and the Poisson algebra \emph{$(X(S,SL(2,\mathbb{C}),\left\{ \,,\,\right\} )$}
of complex functions over them is the algebra of observables. 
\end{theorem}

\subsection{Drinfeld-Turaev Quantization\label{sub:Drinfeld-Turaev-Quantization}}

Now we introduce the ring $\mathbb{C}[[h]]$ of formal polynomials
in $h$ with values in $\mathbb{C}$. This ring has a topological
structure, i.e. for a given power series $a\in\mathbb{C}[[h]]$ the
set $a+h^{n}\mathbb{C}[[h]]$ forms a neighborhood. Now we define 
\begin{definition}
A \emph{Quantization} of a Poisson algebra $P$ is a $\mathbb{C}[[h]]$
algebra $P_{h}$ together with the $\mathbb{C}$-algebra isomorphism
$\Theta:P_{h}/hP\to P$ so that

1. the modul $P_{h}$ is isomorphic to $V[[h]]$ for a $\mathbb{C}$
vector space $V$

2. let $a,b\in P$ and $a',b'\in P_{h}$ be $\Theta(a)=a'$, $\Theta(b)=b'$
then\[
\Theta\left(\frac{a'b'-b'a'}{h}\right)=\left\{ a,b\right\} \]

\end{definition}
One speaks of a deformation of the Poisson algebra by using a deformation
parameter $h$ to get a relation between the Poisson bracket and the
commutator. Therefore we have the problem to find the deformation
of the Poisson algebra $(X(S,SL(2,\mathbb{C})),\left\{ \,,\,\right\} )$.
The solution to this problem can be found via two steps: 
\begin{enumerate}
\item at first find another description of the Poisson algebra by a structure
with one parameter at a special value and 
\item secondly vary this parameter to get the deformation. 
\end{enumerate}
Fortunately both problems were already solved (see \cite{Turaev1989,Turaev1991}).
The solution of the first problem is expressed in the theorem: 
\begin{theorem}
The Skein modul $K_{-1}(S\times[0,1])$ (i.e. $t=-1$) has the structure
of an algebra isomorphic to the Poisson algebra $(X(S,SL(2,\mathbb{C}),\left\{ \,,\,\right\} )$.\emph{
}(see also \cite{BulPrzy:1999,Bullock1999}) 
\end{theorem}
Then we have also the solution of the second problem: 
\begin{theorem}
The skein algebra $K_{t}(S\times[0,1])$ is the quantization of the
Poisson algebra $(X(S,SL(2,\mathbb{C}),\left\{ \,,\,\right\} )$ with
the deformation parameter $t=\exp(h/4)$.(see also \cite{BulPrzy:1999})\emph{ }
\end{theorem}
To understand these solutions we have to introduce the skein module
$K_{t}(M)$ of a 3-manifold $M$ (see \cite{PrasSoss:97}). For that
purpose we consider the set of links $\mathcal{L}(M)$ in $M$ up
to isotopy and construct the vector space $\mathbb{C}\mathcal{L}(M)$
with basis $\mathcal{L}(M)$. Then one can define $\mathbb{C}\mathcal{L}[[t]]$
as ring of formal polynomials having coefficients in $\mathbb{C}\mathcal{L}(M)$.
Now we consider the link diagram of a link, i.e. the projection of
the link to the $\mathbb{R}^{2}$ having the crossings in mind. Choosing
a disk in $\mathbb{R}^{2}$ so that one crossing is inside this disk.
If the three links differ by the three crossings $L_{oo},L_{o},L_{\infty}$
(see figure \ref{fig:skein-crossings}) inside of the disk then these
links are skein related. %
\begin{figure}
\begin{center}\includegraphics[scale=0.25]{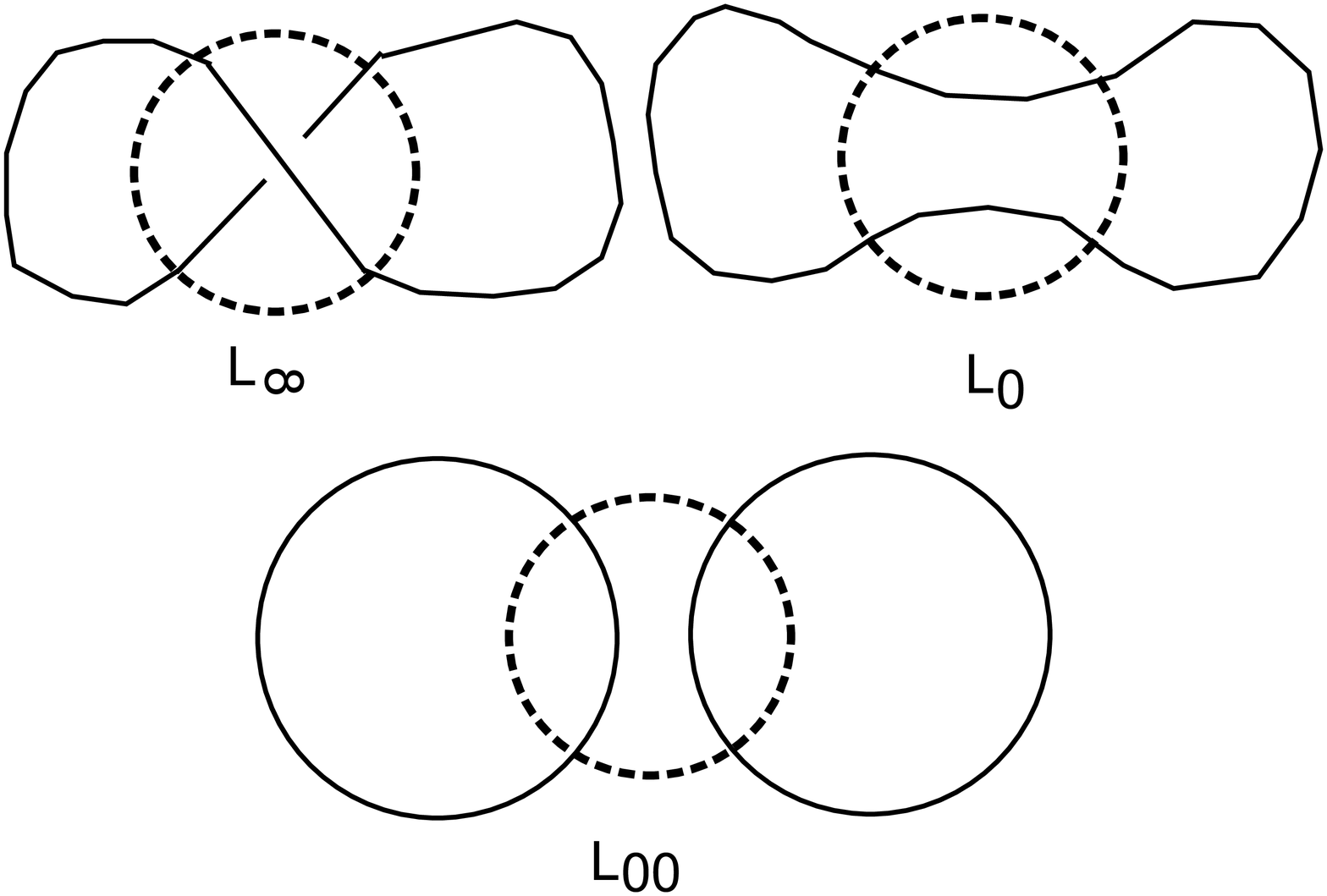}\end{center}

\caption{crossings $L_{\infty},L_{o},L_{oo}$\label{fig:skein-crossings}}

\end{figure}
Then in $\mathbb{C}\mathcal{L}[[t]]$ one writes the skein relation%
\footnote{The relation depends on the group $SL(2,\mathbb{C})$.%
} $L_{\infty}-tL_{o}-t^{-1}L_{oo}$. Furthermore let $L\sqcup O$ be
the disjoint union of the link with a circle then one writes the framing
relation $L\sqcup O+(t^{2}+t^{-2})L$. Let $S(M)$ be the smallest
submodul of $\mathbb{C}\mathcal{L}[[t]]$ containing both relations,
then we define the Kauffman bracket skein modul by $K_{t}(M)=\mathbb{C}\mathcal{L}[[t]]/S(M)$.
We list the following general results about this modul:
\begin{itemize}
\item The modul $K_{-1}(M)$ for $t=-1$ is a commutative algebra.
\item Let $S$ be a surface then $K_{t}(S\times[0,1])$ caries the structure
of an algebra.
\end{itemize}
The algebra structure of $K_{t}(S\times[0,1])$ can be simple seen
by using the diffeomorphism between the sum $S\times[0,1]\cup_{S}S\times[0,1]$
along $S$ and $S\times[0,1]$. Then the product $ab$ of two elements
$a,b\in K_{t}(S\times[0,1])$ is a link in $S\times[0,1]\cup_{S}S\times[0,1]$
corresponding to a link in $S\times[0,1]$ via the diffeomorphism.
The algebra $K_{t}(S\times[0,1])$ is in general non-commutative for
$t\not=-1$. For the following we will omit the interval $[0,1]$
and denote the skein algebra by $K_{t}(S)$. Furthermore we remark,
that \emph{all results remain true if we use an intersection in $L_{\infty}$
instead of a crossing.}

\emph{Ad hoc} the skein algebra is not directly related to the foliation.
We used only the fact that there is an idempotent in the $C^{*}$
algebra represented by a closed curve. It is more satisfying to obtain
a direct relation between both construction. Then the von Neumann
algebra of the foliation is the result of a quantization in the physical
sense. This construction is left for the next subsection.

\subsection{Temperley-Lieb algebra and the operator algebra of the foliation\label{sub:Temperley-Lieb-algebra-foliation}}

In this subsection we will describe a direct relation between the
skein algebra and the factor $I\! I\! I_{1}$ constructed above. At
first we will summarize some of the results above. 
\begin{enumerate}
\item The foliation of the 3-sphere $S^{3}$ has non-trivial Godbillon-Vey
class. The corresponding von Neumann algebra must contain a factor
$I\! I\! I$ algebra.
\item We obtained that the von Neumann algebra is the hyperfinite factor
$I\! I\! I_{1}$ determined by a factor $I\! I_{\infty}$ algebra
via Tomita-Takesaki theory.
\item In the von Neumann algebra there are idempotent operators given by
closed curves in the foliation.
\item The set of closed curves carries the structure of the Poisson algebra
whose quantization is the skein algebra determined by knots and links.
Thus the skein algebra can be seen as a quantization of the fundamental
group.
\end{enumerate}
Thus our main goal in this subsection should be a direct relation
between a suitable skein algebra and the von Neumann algebra of the
foliation. As a first step we remark that a factor $I\! I_{\infty}$
algebra is the tensor product $I\! I_{\infty}=I\! I_{1}\otimes I_{\infty}$.
Thus the main factor is given by the $I\! I_{1}$ factor, i.e. a von
Neumann algebra with finite trace. From the point of view of invariants,
both factors $I\! I_{\infty}$ and $I\! I_{1}$ are Morita-equivalent
leading to the same K-theoretic invariants. 

Now we are faced with the question: Is there any skein algebra isomorphic
to the factor $I\! I_{1}$ algebra? Usually the skein algebra is finite
or finitely generated (as module over the first homology group). Thus
we have to construct a finite algebra reconstructing the factor $I\! I_{1}$
in the limit. Following the theory of Jones \cite{Jon:83}, one uses
a tower of Temperley-Lieb algebras as generated by projection (or
idempotent) operators. Thus, if we are able to show that a skein algebra
constructed from the foliation is isomorphic to the Temperley-Lieb
algebra then we have constructed the factor $I\! I_{1}$ algebra. 

For the construction we go back to factor $I\! I_{\infty}$ foliation
discussed above and identified as the horocycle foliation. Let $P$
the polygon with hyperbolic metric used in subsection \ref{sub:factor-III-case}
and in subsection \ref{sub:Non-cobordant-foliations-S3}. Given a
polygon $P$ as covering space of a surface $S$ (of genus $g>1$)
with non-positive curvature. Denote by $\gamma_{v}$ the geodesic
with initial tangent vector $v$ and by $dist(\gamma_{v}(t),\gamma_{w}(t))$
the distance between two points on two curves. We call the two tangent
vectors $v,w$ of the cover $P$ asymptotic if the distance $dist(\gamma_{v}(t),\gamma_{w}(t))$
is bounded as $t\to\infty$. For a unit tangent vector $v\in T_{1}P$
define the Busemann function $b_{v}:P\to\mathbb{R}$ by\[
b_{v}(q)=\lim_{t\to\infty}\left(dist(\gamma_{v}(t),q)-t\right)\]
This function is differentiable and the gradient $-\nabla_{q}b_{v}$
is the unique vector at $q$ asymptotic to $v$. We define alternatively
the horocycle $h(v)$ (determined by $v$) as the level set $b_{v}^{-1}(0)$.
Clearly $h(v)$ is the limit as $R\to\infty$ of the geodesic circles
of radius $R$ centered at $\gamma_{v}(R)$. Let $W(v)$ be the set
of vectors $w$ asymptotic to $v$ with footpoints on $h(v)$ (see
Fig. \ref{fig:horocycle-fig}), i.e.\[
W(v)=\left\{ -\nabla_{q}b_{v}\,|\, q\in h(v)\right\} \:.\]
The curves $W(v),\: v\in T_{1}P$ are the leaves of the horocycle
foliation $W$ of $T_{1}P$ which can be lifted to a horocycle foliation
$W$ on $T_{1}S$. Remember a horocycle is a circle in the interior
of $P$ touching the boundary at one point. Now we consider the flow
in $T_{1}P$ along a horocycle with unit speed which induces a codimension-1
foliation in $T_{1}P$. The horocycle foliation is parametrized by
the set of horocycles on $P$. Thus the set of unit tangent vectors
labels the leaves of the foliation or the leaf space is parametrized
by unit tangent vectors. Furthermore we remark that every horocycle
is also determined by a unit tangent vector. By definition, the set
of unit tangent vectors is completely determined by curves in $P$.
Every horocycle meets the boundary of $P$ at one point, which we
mark (see Fig. \ref{fig:horocycle}), say $m_{1},\ldots,m_{n}$. %
\begin{figure}
\begin{center}

\includegraphics[scale=0.25]{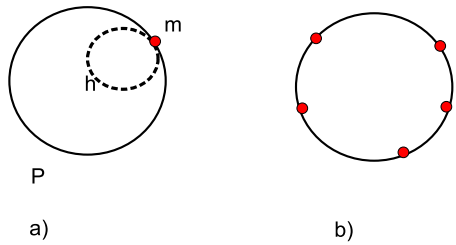}

\end{center}

\caption{fig. a) an example for a horocycle $h$ on $P$ and fig. b) marked
points on the boundary of the polygon $P$\label{fig:horocycle}}

\end{figure}
\begin{figure}
\begin{center}

\includegraphics[scale=0.25]{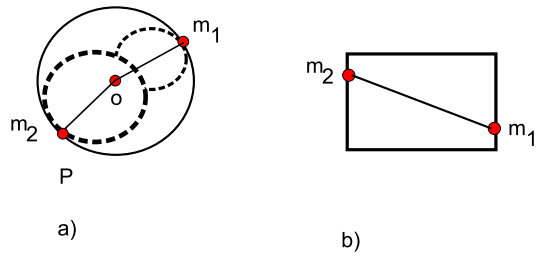}

\end{center}

\caption{a) flow line between two marked points defined via two horocycles,
b) simple picture as substitute\label{fig:flow-line}}

\end{figure}
Then by uniqueness of the flow, there is a curve from the boundary
point $m_{1}$ in the interior of $P$ meeting a point $o$ followed
by a curve from this point $o$ to another boundary point $m_{2}$
(see figure \ref{fig:flow-line}). Thus in general we obtain curves
in $P$ going from one marked boundary point $m_{k}$ to another marked
boundary point $m_{l}$. We need only a countable number of these
points. Therefore we choose the number of vertices $k$ of the polygon
$P$. Without loss of generality we choose an even number of vertices
$k=2n$. Then any pair of vertices is connected by one geodesic path.
To express the grouping of the marked points, we used a rectangle
instead of the circle (as indicated in Fig. \ref{fig:flow-line}).
All other paths can be generated by a simple variation of the start
and end point (isotopy). Using the horocycles we obtain a flow for
every pair of marked points. The set of unit tangent vectors labels
the leaves of the foliation and can be described by curves between
the marked points. Then we group the marked points and assume that
we have the same number of marked points on the left and on the right
side of $P$. Now we have to define the (formal) sum of two flows.
A flow starts on one marked point of one side going to one point at
the other side (see figure \ref{fig:flow-line}). By using this definition
we obtain also singularities, i.e. crossings of flows. But the singularities
or intersection points can be solved to get non-singular flows.%
\begin{figure}
\begin{center}

\includegraphics[scale=0.25]{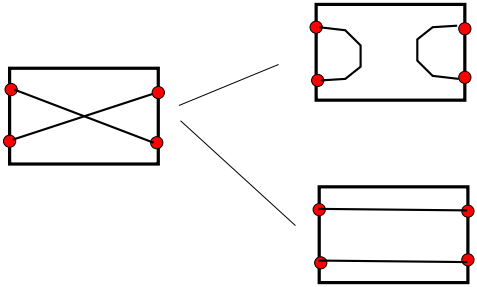}

\end{center}

\caption{resolution of the flow singularities\label{fig:resolution-sing}}

\end{figure}
 The figure \ref{fig:resolution-sing} shows the method%
\footnote{The method was used in the theory of finite knot invariants (Vassiliev
invariants) and is known as STU relation.%
}. In the subsection \ref{sub:Drinfeld-Turaev-Quantization} we introduced
the skein algebra. We define the resolution of the singularity together
with a parameter $t$ in similarity to the Skein relation $L_{\infty}=tL_{o}+t^{-1}L_{oo}$.
By this method we are also able to define a sum of two flows by reversing
the procedure: the sum of two flows will produce a singular flow.
Now we consider two polygons with the same number of marked points
on one side. These polygons can be put together (see figure \ref{fig:product-structure})
to define a product. %
\begin{figure}
\begin{center}

\includegraphics[scale=0.25]{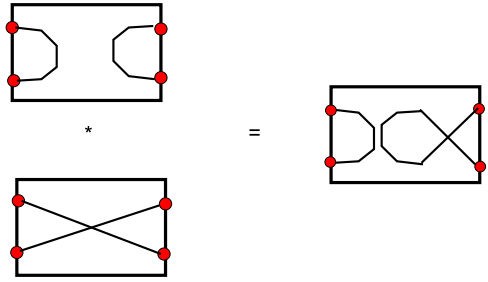}

\end{center}

\caption{product structure \label{fig:product-structure}}

\end{figure}
 Before presenting the theorem, we summarize the definition of the
operations:
\begin{enumerate}
\item Given a polygon $P(k,k)$ with $2k$ marked points, i.e. with $k$
marked points of each side. There is only one internal line between
two marked points, i.e. $2k$ marked points are connected by $k$
lines. Two internal lines do not intersect. The marked points of each
side have equal distances to each other. One can add or remove an
internal circle.
\item \emph{Product:} The connected sum of polygons $P(k,k)$ with $k$
marked points on each side is the product in the algebra (see Fig.
\ref{fig:product-structure}).
\item \emph{Multiplication by number}: The change in the distance between
two marked points is the multiplication with a real number. The details
of this operation is not important at the moment.
\item \emph{Sum}: A linear combination between two polygons is represented
by the crossing of two internal lines (see Fig. \ref{fig:resolution-sing}).
\item $*$\emph{operation}: A $180\text{\textdegree}$ rotation of the polygon
$P^{*}(k,k)$ is the $*$operation.
\end{enumerate}
Putting all these definitions together we obtain:
\begin{theorem}
\label{thm:leaf-space-Temperley-lieb-algebra}The leaf space of the
horocycle foliation of a surface $S$ (of genus $g>1$) is represented
by the leaf space of a horocycle foliation of a polygon $P$ in $\mathbb{H}^{2}$.
If one extends the leaf space to allow (countable many) crossings
between two leaves and considers a connected sum of the polygons then
this extended leaf space admits the structure of an $*-$algebra (see
the rules above). Let $P(k,k)$ be a polygon with $2k$ vertices generating
a foliation (see subsection \ref{sub:Non-cobordant-foliations-S3})
with $*-$algebra $P_{k}$ as extended leaf space. This algebra is
isomorphic to the Temperley-Lieb-algebra $TL_{k}$ generated by $k$
elements $\left\{ e_{1},\ldots,e_{n}\right\} $ subject to the relations
(\ref{Jones-algebra}). The generators are idempotent operators represented
by closed curves in $S$. In the direct limit one obtains the hyperfinite
factor $I\! I_{1}$ algebra given as tower of Temperley-Lieb algebras
(and equal to the skein algebra of a marked disk).\end{theorem}
\begin{proof}
The relation between the two foliations was already shown above. The
extended leaf space was also constructed above. Especially we mention
the operations. The polygon $P(k,k)$, where each internal line connects
one marked point on one side with a marked point on the other side,
is the identity in the algebra. The next complicated case is given
by a polygon where two marked points on each side are connected by
one internal line. We denote these polygons by $e_{n}$ and refer
to figure \ref{fig:generators-TL-alg} for the convention. The product
operation is defined above. %
\begin{figure}
\begin{center}

\includegraphics[scale=0.25]{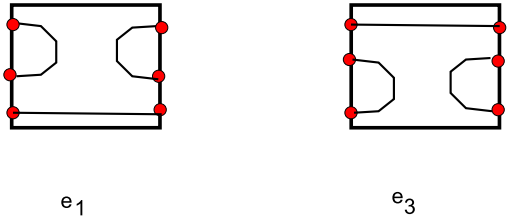}

\end{center}

\caption{example for generators $e_{1}$ and $e_{3}$ of the Temperley-Lieb
algebra \label{fig:generators-TL-alg}}

\end{figure}
 By simple graphical manipulations (see \cite{PrasSoss:97} $\S$26.9.)
using our rules above, we obtain the relations:\begin{eqnarray}
e_{i}^{2}=e_{i}\,, & e_{i}e_{j}=e_{j}e_{i}\,:\,|i-j|>1,\nonumber \\
e_{i}e_{i+1}e_{i}=e_{i}\,, & e_{i+1}e_{i}e_{i+1}=e_{i+1}\,,\, e_{i}^{*}=e_{i}\label{Jones-algebra}\end{eqnarray}
where the $*$operation obviously do not change the generators $e_{i}$.
This algebra is the Temperley-Lieb algebra $TL_{n}$ generated by
$\left\{ e_{1},\ldots,e_{n}\right\} $. The inclusion $TL_{n}\subset TL_{n+1}$
is given by adding two marked points, one point on each side connected
by one internal line. As Jones \cite{Jon:83} showed: the limit case
$\lim_{n\to\infty}TL_{n}$ (considered as direct limit) is the factor
$I\! I_{1}$. Thus we have constructed the factor $I\! I_{1}$ algebra
as skein algebra. $\square$
\end{proof}
We will finish this subsection with one remark. In the factor Temperley-Lieb
algebra there is an unique idempotent operator, the Jones-Wenzl idempotent,
which is related to Connes idempotent operator in the operator algebra
of the foliation by our construction.

\section{Knots as states\label{sec:Knots-as-states} }

Lets summarize the situation again:
\begin{enumerate}
\item A small exotic $\mathbb{R}^{4}$ is related to a codimension-1 foliation
of the 3-sphere $S^{3}$ unique up to foliated cobordism.
\item The leaf space of this foliation is a factor $I\! I\! I_{1}$ algebra
related to a foliation with leaf space a factor $I\! I_{\infty}=I\! I_{1}\otimes I_{\infty}$
algebra.
\item The algebra $I\! I_{1}$ is the direct limit of the Temperley-Lieb
algebras $TL_{n}$.
\item Closed curves transversal to the foliation are related to idempotent
operators of the leaf space.
\end{enumerate}
This section is two-fold. On one side we will try to construct the
states in the factor $I\! I\! I_{1}$ algebra by geometric methods.
But on the other side we will also discuss the operators interpreted
as cobordisms between the states (in the sense of the topological
QFT a la Atiyah). This section is different from the others, we present
only the ideas and leave the proofs for further papers.

\subsection{Observables as closed curves}

In subsection \ref{sub:The-observable-algebra} we described the observable
algebra (i.e. the Poisson algebra) as the space of flat connections
$X(S,SL(2,\mathbb{C}))$ where the observables are the holonomies
(\ref{eq:holonomy-definition}) of the flat connection along closed
curves. This picture remains true after quantization where we obtain
the skein algebra. In subsection \ref{sub:Idempotent-operators-knot-cobordism}
we considered the idempotent operators represented by closed curves.
Finally we showed in Theorem \ref{thm:leaf-space-Temperley-lieb-algebra}
that these closed curves generate the Temperley-Lieb algebra. Especially
the linear combination of two generators is represented by a singular
flow (see Fig. \ref{fig:resolution-sing}). An observable is represented
by a hermitean operator or equivalently by a linear combination of
projectors. Therefore an observable is given by a singular flow or
by using Theorem \ref{thm:idempotent-operators-as-closed-curves}
as a singular knot (or link). 

Now we will consider a special idempotent operator known as Jones-Wenzl
idempotent. For its construction we have to modify the definition
(\ref{Jones-algebra}) of the algebra: instead of $e_{i}^{2}=e_{i}$
we write \[
e_{i}^{2}=\tau e_{i}\]
where $\tau$ is a real number given by a closed circle in the polygon
$P(k,k)$. So, we modify one rule in the definition of our algebra:
adding an internal circle is equivalent to a multiplication with $\tau$.
If $\tau$ is the number $\tau=a_{0}^{2}+a_{0}^{-2}$ with $a_{0}$
a $4n$th root of unity ($a_{0}^{4k}\not=1$ for $k=1,\ldots,n-1$)
and $A_{n}\subset TL_{n}$ a subalgebra generated of $\left\{ e_{1},\ldots,e_{n}\right\} $
missing the identity $1_{n}$ then there is an element $f^{(n)}$
with \begin{eqnarray*}
f^{(n)}A_{n} & = & A_{n}f^{(n)}=0\\
1_{n}-f^{(n)} & \in & A_{n}\\
f^{(n)}f^{(n)} & = & f^{(n)}\end{eqnarray*}
called the Jones-Wenzl idempotent \cite{PrasSoss:97}. This idempotent
is used to define a 3-manifold invariant of Witten type \cite{Wit:89}.
Witten defines this invariant by the state sum of the Chern-Simons
theory (for a suitable gauge group, here it is $SU(2)$). This unexpected
relation gives a hint for a possible action related to our QFT (given
by the factor $I\! I\! I_{1}$).

\subsection{Knot concordance and capped gropes}

In the previous subsection we considered the observables of the theory
given by singular knots or links. Now we are interested in the construction
of states in the factor $I\! I_{1}-$algebra $A$, i.e. linear functionals
$f:A\to\mathbb{C}$ which are positive ($f(x^{*}x)\geq0$ for all
$x\in A$) and normed $||f||=1$. Usually one constructs the states
by a representation (GNS) of the algebra $A$ into a Hilbert space.
From the physical point of view, one can argue that every pure state
must be corresponding to a classical state. Because of the relation
between the Poisson algebra $X(S,SL(2,\mathbb{C}))$ and its quantization
as skein space $K_{t}(S\times[0,1])$, the state must be the holonomy
along a knot. To see this fact, we will follow another path. Our theory
is purely topological, i.e. we assign to every observable (as endomorphism
of some Hilbert space) a singular knot. Following the axioms of a
Topological QFT (TQFT) by Atiyah, one assigns to a cylinder like $S\times[0,1]$
an endomorphism. In our case it is an element of the skein space,
i.e. a singular knot. Singular knots appear in the theory of Vassiliev
invariants where one considers transitions between two different knots.
For our approach we have to interpret the 3-dimensional objects in
the 4-dimensional context.

For that purpose we consider the horocycle foliation of the unit tangent
bundle $T_{1}S$ of a surface $S$ with genus $g>1$. As usual we
associate to this foliation a codimension-1 foliation of the 3-sphere
given by a polygon $P$ in $\mathbb{H}^{2}$ with $4g$ vertices (or
sides). In the proof of the Theorem \ref{thm:codim-1-foli-radial-fam},
the cobordism class of this foliation (i.e. the Godbillon-Vey invariant)
is determined by a capped grope (with its design). This capped grope
is determined by a path in a binary tree (a trivalent tree), see subsection
\ref{sub:design-of-CH}. The capped grope is embedded in some 4-space,
i.e. the boundary of the capped grope, a circle, is embedded in the
3-space ($S^{3}$ or $\mathbb{R}^{3}$). The caps in the capped grope
are also immersed disks (see subsection \ref{sub:Capped-gropes}).
Therefore, by fixing one cap, we have a cobordism between two embedded
circles or a cobordism between two knots. Especially, this cobordism
is induced from a cobordism between the embedding spaces, usually
called a knot concordance. Therefore we have identified the endomorphisms
as knot concordance between two knots, the states. There is an extensive
literature for the relation between knot concordances and capped gropes
\cite{TeichnerOrrCochran2003,TeichnerConant2004}. Especially we mention
the relation to Vassiliev invariants \cite{TeicherConant2004} and
the Kontsevich integral \cite{TeichnerScheideman2004}. Here we presented
only this relation and refer to our future work.

\section{Discussion}

In this paper we presented a variety of relations between codimension-1
foliations of the 3-sphere $S^{3}$ and noncommutative algebras. By
using the results of our previous paper \cite{AsselmeyerKrol2009},
we obtain a relation between (small) exotic smoothness of the $\mathbb{R}^{4}$
and noncommutativity via the noncommutative leaf space of the foliation
and the Casson handle. Thus we get our main result of this paper:
\\
\emph{The Casson handle carries the structure of a noncommutative
space determined by a factor $I\! I_{1}$ algebra which is related
to the skein algebra of the disk with marked points and to the leaf
space of the horocycle foliation.} \\
Thus we have obtained a direct link between noncommutative spaces
and exotic 4-manifolds which can be used to get a direct relation
to quantum field theory. One of the central elements in the algebraic
quantum field theory is the Tomita-Takesaki theory leading to the
$I\! I\! I_{1}$ factor as vacuum sector \cite{Borchers2000}. As
a possible candidate one has loop quantum gravity with an unique diffeomorphism-invariant
vacuum state \cite{Thiemann2006}. Especially the relation to skein
spaces and knot concordances are very attractive for future work.
Our work has also some overlap with the nice work \cite{BertozziniConti2010}
for quantum gravity. We will close our paper with some speculations
for a possible interpretation of the capped gropes as the trees in
Connes-Kreimer renormalization theory. Starting point is the observation,
that the Hopf algebra of formal vector fields in a codimension-1 foliation
is isomorphic to the Hopf algebra of renormalization in QFT a la Kreimer.
In our context it means that exotic smoothness (as described by codimension-1
foliations) has many to do with renormalization in QFT. Again we refer
to our future work.

\appendix

\section*{Appendix}

\section{Proof of Theorem \ref{thm:codim-1-foli-radial-fam}}
\begin{proof}
We consider a tubular neighborhood $\partial A\times[0,1]\subset\mathbb{R}_{1}^{4}$
of $\partial A$ and glue the Casson handle along some 2-handle. Now
we will weaken the Casson handles by using capped gropes (see chapter
1-4 in \cite{FreQui:90}) denoted by GCH. These differ from Casson
handles in that many surface stages are interspersed between the immersed
disks of Casson's construction. The GCH are also indexed by rooted
finitely branching objects. The growth rate of their stages was determined
in \cite{AncelStarbird1989} (Theorem A) to be at least exponential
(more than $2^{n}$). In the proof of Theorem 3.2 in \cite{DeMichFreedman1992}
the gaps in the design where used. The gaps in case of the Casson
handle are not manifolds and look like $S^{1}\times D^{3}/Wh$. In
case of the capped grope one has {}``good'' gaps of the form $S^{1}\times D^{3}$.
That is the reason why we switch to these objects now. Now we decompose
the gap by $gap=S^{1}\times D^{3}=S^{1}\times D^{2}\times I$ with
the unit interval $I=[0,1]=D^{1}$. The boundary is a decomposition
$\partial(gap)=(S^{1}\times S^{1}\times I)\cup(S^{1}\times D^{2}\times\left\{ 0,1\right\} )$
of the caps (north and south) and the equator region (see Fig. \ref{fig:decomposition-of-gap}).%
\begin{figure}
\includegraphics[scale=0.5]{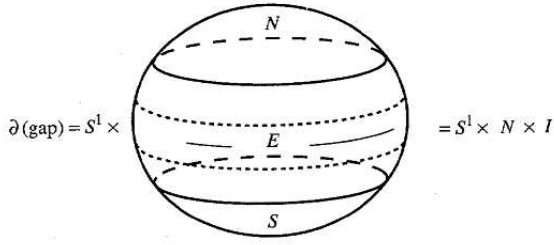}

\caption{decomposition of the boundary of a gap (fig. repainted from \cite{DeMichFreedman1992})\label{fig:decomposition-of-gap}}

\end{figure}
 The radius coordinate $\rho$ defined above is identified with the
unit interval of the $gap$ (see the proof of Theorem 3.2). In the
notation of \cite{DeMichFreedman1992}, we think of each gap as $gap=S^{1}\times N\times I$
where $N=D^{2}$ is the neighborhood around the north pole of the
2-sphere in Fig. \ref{fig:decomposition-of-gap}. Using the reembedding
theorems every GCH embeds in the open 2-handle and induces a foliation
visualized in Fig. %
\begin{figure}
\includegraphics[scale=0.5]{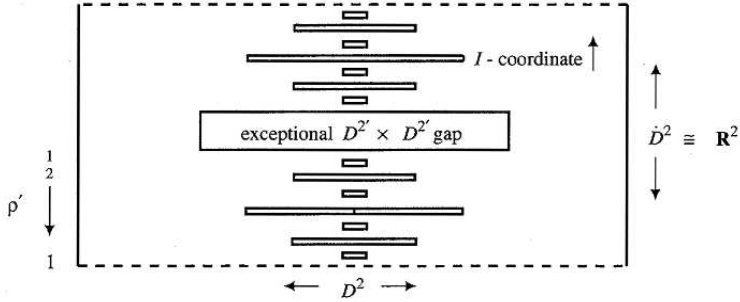}

\caption{visualization of foliation as coordinatization of the design (fig.
repainted from \cite{DeMichFreedman1992})\label{fig:visualization-of-foliation-design}}

\end{figure}
. As described in subsection \ref{sub:design-of-CH} the simplest
tree $Tree_{+}$ belongs to the binary sequence $.000\ldots$ and
is represented by $t=0$ and the radius $r=1/(1-t)=+\infty$. The
foliation of the design is perpendicular to $S^{1}\times N$, i.e.
$S^{1}\times\{longitudes\}$ are the leaves. The intersection of the
leaves with $S^{1}\times N$ produces a foliation of the disk $N$.
This disk is given up to conformal automorphism by fixing the sphere
$S^{2}\supset N$, i.e. the disk is invariant w.r.t. the group $PSL(2,\mathbb{R})$.
The boundary of $N$ is given by geodesic curves. The $PSL(2,\mathbb{R})-$invariance
induces a mapping of the disk $N$ into the hyperbolic space $\mathbb{H}^{2}$,
where every $PSL(2,\mathbb{R})$ transformation is an isometry now.
Then the foliated $N$ is mapped to a foliated polygon $P$ in $\mathbb{H}^{2}$,
where the foliation is $PSL(2,\mathbb{R})-$invariant. From this point
of view we interpret $S^{1}\times N$ as the unit tangent bundle of
the polygon $T_{1}P$. Then the volume of the polygon $P$ is the
volume of the disk $N$, i.e. $vol(P)=vol(N)$ and we choose the number
of vertices of $P$ in a suitable manner by defining the geodesic
arcs forming the boundary of $N$. As Fig. \ref{fig:decomposition-of-gap}
indicated, the disk $N$ is also part of the boundary $\partial(gap)=S^{1}\times S^{2}$
of the gap. Then the unit interval in the gap is directly related
to the radius $r$ of the 2-sphere $S^{2}\supset N$ and this radius
determines the volume of the disk $N$ (as part of the upper hemisphere
of $S^{2}$, see Fig. \ref{fig:decomposition-of-gap}). But then by
using $PSL(2,\mathbb{R})-$invariance, we obtain the relation $vol(P)=r^{2}$.
\end{proof}
The tubular neighborhood $\partial A\times[0,1]=\partial A\times I$
can be chosen in such a manner that the coordinate $\rho$ agrees
with the unit interval $I$ of the neighborhood. In subsection \ref{sub:design-of-CH}
we described the foliation of the design as a foliated cobordism between
two disks $N,S$ given by the 2-sphere $S^{2}$ (see Fig. \ref{fig:foliation-of-S2}).
As described above, every disk ($N$ or $S$) is related to the polygon
$P$ (without loss of generality we use $N$ and $S$ with the same
volume $vol(N)=vol(S)$). Then the foliation of the design induces
a foliation of the cobordism $\partial A\times I$. The space $D^{2}\times S^{2}$
(the 5-stage towers) is the leaf of the foliated cobordism transverse
to the foliation on the boundary $\partial A$. Then the restriction
on the boundary $\partial A\times\left\{ 0,1\right\} $ induces a
foliation on $\partial A$ determined by the volume $vol(P)$ of the
polygon. So, the radius $r^{2}$ (proportional to the volume of $vol(P)$)
is a cobordism invariant of the foliation. By Theorem \ref{thm:foliation-3MF}
we obtain a codimension-1 foliation of $\partial A$ induced from
a foliation of the $S^{3}$. As shown \cite{Thu:72} (see also the
book \cite{Tamura1992} chapter VIII for the details) this invariant
agrees with the Godbillon-Vey invariant $GV=r^{2}$. Then two non-diffeomorphic
small exotic $\mathbb{R}^{4}$ (for $s\not=t$) have different radial
coordinates ($\frac{1}{1-s}=r_{s}\not=r_{t}=\frac{1}{1-t}$) and therefore
different Godbillon-Vey invariants $r_{s}^{2}\not=r_{t}^{2}$. The
corresponding foliations are non-cobordant to each other. $\square$

\section*{Acknowledgment}

T.A. wants to thank C.H. Brans and H. Ros\'e for numerous discussions
over the years about the relation of exotic smoothness to physics.
J.K. benefited much from the explanations given to him by Robert Gompf
regarding 4-smoothness several years ago, and discussions with Jan
S{\l{}}adkowski.


\end{document}